\documentclass[letterpaper,journal]{IEEEtran}
\usepackage{amsmath,amsfonts}
\usepackage{amssymb}
\usepackage{amsthm}
\usepackage{cite}
\usepackage{footnote}
\usepackage{stfloats}
\usepackage{array}
\usepackage{enumerate}
\usepackage{multirow}
\usepackage{balance}
\usepackage{algorithm}
\usepackage{algpseudocode}
\usepackage{graphicx}
\usepackage{subfigure}
\usepackage{textcomp}
\usepackage{color}
\usepackage{xcolor}
\usepackage{tabularx}
\usepackage{booktabs}
\usepackage{verbatim}
\usepackage{hyperref}
\usepackage{makecell}

\theoremstyle{definition}
\newtheorem{theorem}{Theorem}
\newtheorem{lemma}{Lemma}
\newtheorem{proposition}{Proposition}
\newtheorem{corollary}{Corollary}

\begin{document}


\title{Covert Digital Semantic Communication With Movable Antennas: Constellation Shaping, Spatial Optimization, and Semantic Mapping}

\author{Yunfan~Bai,~\IEEEmembership{Graduate Student Member,~IEEE}, Yuwen~Qian, Hao~Jiang,~\IEEEmembership{Senior Member,~IEEE}, Yiyang~Ni,~\IEEEmembership{Senior Member,~IEEE}, Feng~Shu,~\IEEEmembership{Senior Member,~IEEE}, Jun~Li,~\IEEEmembership{Fellow,~IEEE}, and Shuning~Zhang

\thanks{Yunfan~Bai, Yuwen~Qian, and Shuning~Zhang are with the School of Electronic and Optical Engineering, Nanjing University of Science and Technology, Nanjing 210094, China (e-mail:\{baiyunfan, admon, and zhangshuning\}@njust.edu.cn).}

\thanks{Hao~Jiang is with the School of Cyber Science and Engineering, Southeast University, Nanjing 210096, China. (e-mail: jiang.hao@seu.edu.cn).}

\thanks{Yiyang~Ni is with the School of Computer Engineering, Jiangsu Second Normal University, Nanjing 210037, China (e-mail:niyy@jssnu.edu.cn ).}

\thanks{Feng~Shu is with the School of Information and Communication Engineering and Collaborative Innovation Center of Information Technology, Hainan University, Haikou 570228, China, and also with the School of Electronic and Optical Engineering, Nanjing University of Science and Technology, Nanjing 210094, China. (e-mail: shufeng0101@163.com).}

\thanks{Jun~Li is with the School of Information Science and Engineering, Southeast University, Nanjing 210096, China. (e-mail: jun.li@seu.edu.cn).}

}

\maketitle

\begin{abstract}
Digital semantic communication~(SemCom) is emerging as a bandwidth-efficient paradigm for next-generation wireless networks by delivering the meaning of a source over finite constellations.
However, the broadcast nature of wireless channels exposes transmissions to adversarial wardens, while discrete symbols deviate from the Gaussian signaling assumed by conventional covert communication.
In this paper, we propose a covert digital SemCom system assisted by movable antennas~(MAs), where the constellation, the antenna positions, and the semantic transceiver are jointly designed to conceal the transmission while preserving the semantic fidelity.
Specifically, the maximum covert rate against the optimal detector with unknown noise level is derived in closed-form, which reveals that the legitimate signal-to-noise ratio~(SNR) scales with the moment-matching order of the constellation and the spatial gain provided by the MA array.
A Gauss--Laguerre ring constellation is then designed to raise the moment-matching order under the peak power constraints, and a covert symbol mapping with rank-to-ring assignment and scrambling is developed to carry the semantics in the ordering of the latent magnitudes.
Guided by the derived semantic distortion bound, the semantic encoder and decoder are trained end-to-end through an invertible flow over the designed constellation and MA array.
Simulation results demonstrate that the proposed scheme supports a higher achievable covert rate and improves the reconstruction quality by 24\% compared to state-of-the-art SemCom schemes.

\end{abstract}


\section{Introduction}
\vspace{-0.3em}

\IEEEPARstart{T}{he} next-generation wireless networks are envisioned to support efficient information exchange and reliable task execution for intelligent services under spectrum constraints.
To reconcile service demands with limited transmission resources, semantic communication~(SemCom) prioritizes task-relevant information over exact bit recovery to reduce the communication overhead~\cite{xie2021deep}.
Reliable semantic delivery over noisy channels has motivated deep learning~(DL)-based joint source-channel coding~(JSCC), which jointly optimizes source representation and channel protection~\cite{nguyen2026contemporary}.
To meet the requirement of practical transceivers that modulate finite constellations under standard hardware constraints, digital SemCom has emerged as a promising architecture for semantic transmission~\cite{zou2025analogdigital}.

Beyond transmission reliability, SemCom security is critical since open wireless channels expose sensitive information to unauthorized interception~\cite{guo2025survey}.
Despite the protection of message content, observable transmission activity reveals the presence of communicating devices to unintended observers.
Covert communication complements content protection through a low probability of detection (LPD), which limits a passive warden's ability to distinguish transmission from silence.
The square-root law establishes a fundamental throughput limit for covert communication under the classical additive white Gaussian noise~(AWGN) channel model~\cite{bash2013limits}.
Under the stringent constraint of transmission budget, compact semantic representations support reconstruction and inference with reduced communication overhead~\cite{bai2026mimo}.
Reliable delivery of compact semantic representations under limited transmit power motivates spatial channel adaptation to improve the legitimate channel gain.
Meanwhile, movable antennas~(MAs) reshape the wireless channel through position adjustment within a confined region, which improves the legitimate signal-to-noise ratio~(SNR) without increasing the number of radio-frequency chains~\cite{zhu2024modeling,zhu2023movable}.
The combination of semantic compactness and spatial reconfiguration therefore offers a promising route to covert semantic delivery over digital transceivers.

Recent advances in MA-assisted transmission have optimized the antenna positions under statistical channel state information~(CSI) to enable infrequent antenna reconfiguration~\cite{yan2026statistical}.
For covert transmission, the position-dependent channel response shapes the legitimate reception and warden observations simultaneously.
Joint antenna positioning and beamforming were designed to separate a warden located close to the legitimate receiver with wavelength-scale spatial resolution~\cite{wang2025resolution}.
Movable frequency diverse arrays further extended the spatial design to multiple wardens by adjusting the antenna positions and the carrier frequencies~\cite{cheng2025movablefda}.
With a full-duplex receiver, the antenna positions and the artificial noise~(AN) power were jointly optimized to improve the covert throughput~\cite{hu2026movable}.
Beyond the spatial configuration, the reflection coefficients of an intelligent reflecting surface~(IRS) were optimized together with the transmit power for finite-alphabet inputs~\cite{zhou2025finitealphabet}.
However, the above MA-assisted designs adopt Gaussian signaling under which warden observations are characterized by the received power, whereas the finite-alphabet analysis fixes the constellation over fixed antennas.
Consequently, formal design criteria remain lacking for the admissible transmit power and antenna positions, and existing bounds yield overly conservative limits when digital transmitters use finite constellations under warden noise uncertainty.

The symbol statistics are determined by the SemCom transceiver, where reconstruction-oriented JSCC has advanced from convolutional to transformer backbones and SwinJSCC adapts continuous channel symbols to channel conditions and transmission rates within a shared model~\cite{yang2025swinjscc}.
Continuous symbols are not directly compatible with digital hardware, and DeepJSCC-Q therefore constrains the encoder output to a finite constellation through differentiable quantization~\cite{tung2022deepjsccq}.
A masked vector-quantized variational autoencoder~(VQ-VAE) further transmits discrete codebook indices for robust semantic delivery~\cite{hu2023robust}.
For resource adaptation, Digital-SC integrates learned nonlinear quantization with an adaptive network split~\cite{guo2025digitalsc}.
At the modulation level, joint coding-modulation learns source-dependent probabilities over the constellation points through a variational autoencoder~\cite{bo2024joint}.
Universal JSCC supports different modulation orders through shared neural transforms and learned codebooks~\cite{huh2025universal}.
Joint semantic-channel coding and modulation further adapts digital tokens to the semantic content and the channel condition~\cite{ying2026token}.
Since the discrete operations hinder gradient-based training, an alternating strategy with mask attacks emulates bit errors across the nondifferentiable modules~\cite{gong2026digital}.
Nevertheless, the above designs let the symbol probabilities follow the source and learn the mappings for fidelity, which leaves the warden with source-dependent correlated symbols and rules out a covertness guarantee without a mapping that fixes the block distribution.


Toward integrating covertness with semantics, covertness-oriented AN was designed to conceal sensitive image information while the semantic features can be recovered at the legitimate receiver~\cite{yin2026semantic}.
A multiple-input multiple-output~(MIMO) framework conveys covert semantics through spatial modulation under AN masking, where an excess-distortion exponent bound guides the training of the semantic network~\cite{bai2026mimo}.
Under noise uncertainty at the warden, frequency hopping jointly optimizes the transmit power and the frequency allocation to increase the detection error~\cite{li2026frequencyhopping}.
The integration of semantic compression with frequency hopping further optimizes the reconstruction quality and the transmission latency under finite-blocklength covertness constraints~\cite{hui2026frequencyhopping}.
Although existing semantic covert schemes improve the recovery quality under covertness, the joint influence of MA placement and finite-constellation mappings on attainable semantic distortion remains insufficiently characterized under constraints on covertness and block statistics.


To address these issues, we propose a digital semantic covert communication framework assisted by MAs.
In particular, the main contributions of this paper are listed as follows.

\begin{itemize}
    \item We propose a covert digital SemCom framework assisted by MAs, where the semantic latent is conveyed over a finite constellation in the presence of a warden with unknown noise level. 
    Specifically, the constellation, the antenna positions and the semantic transceiver are jointly designed to minimize the semantic distortion under the covertness constraint.

    \item To characterize the covertness of finite constellations, the maximum covert rate under the composite null hypothesis is derived in closed-form via the Kullback--Leibler~(KL) divergence. 
    A Gauss--Laguerre~(GL) ring constellation is then designed to raise the moment matching order under the peak power constraint.

    \item To exploit the spatial degrees of freedom provided by MAs, a matching-order-aware placement is derived to maximize the legitimate SNR advantage over the warden. 
    A differentiable outage surrogate is developed to characterize the covert power loss, which couples the antenna positions with the constellation.

    \item A digital semantic transceiver is developed over the designed constellation and MA array, where a covert symbol mapping with rank-to-ring assignment and scrambling decouples the transmitted symbols from the source. 
    By leveraging a derived semantic distortion bound, an invertible flow together with the semantic encoder and decoder network is trained end-to-end to recover high-fidelity semantics under the covertness constraint.
\end{itemize}



\textit{Notations:} Boldface lowercase and uppercase letters denote vectors and matrices, and $(\cdot)^{\top}$, $(\cdot)^{H}$, $(\cdot)^{*}$ and $(\cdot)^{\star}$ denote the transpose, the conjugate transpose, the complex conjugate and an optimal value.
$\mathbb{E}[\cdot]$, $\Pr\{\cdot\}$, $D(\cdot\|\cdot)$, $I(\cdot;\cdot)$ and $h(\cdot)$ denote the expectation, the probability, the Kullback--Leibler divergence, the mutual information and the differential entropy, where all information quantities are in nats with the natural logarithm.
$\mathcal{CN}(\mathbf{0},\boldsymbol{\Sigma})$ denotes the circularly symmetric complex Gaussian distribution with covariance $\boldsymbol{\Sigma}$, $\lambda_{\max}(\cdot)$ denotes the largest eigenvalue of a Hermitian matrix, $\mu^{\otimes L}$ denotes the $L$-fold product of $\mu$, and $f\sim g$ means that $f/g$ tends to one in the stated limit.


\section{System Model}\label{sec: system_model}
\vspace{-0.3em}

As illustrated in Fig.~\ref{system_model}, the considered covert SemCom system consists of a legitimate transmitter Alice with $M$ movable antennas, a legitimate receiver Bob with $N$ fixed antennas and a warden Willie with $M_{w}$ fixed antennas.
Alice delivers a sequence of semantic objects to Bob and simultaneously keeps the transmission undetectable to Willie, where each covert session spans $n$ channel uses and carries $J=n/L$ semantic objects of $L$ symbols each.

\begin{figure}[t]
    \centering
    \includegraphics[width=0.9\linewidth]{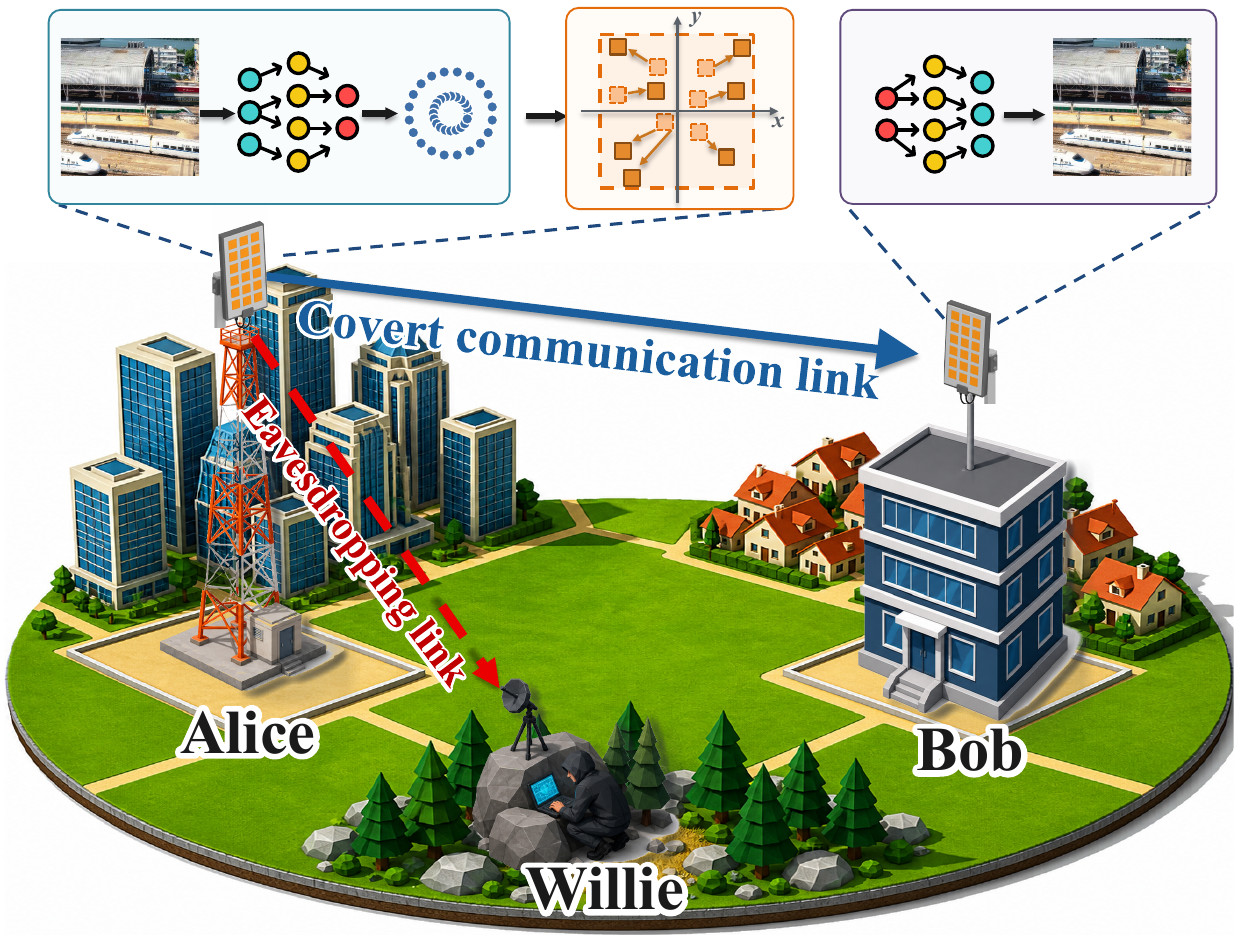}
    \vspace{-0.5em}
    \caption{System model of the proposed MA-assisted digital semantic covert communication framework, where Alice transmits finite-constellation semantic symbols to Bob in the presence of the warden Willie.
}
    \label{system_model}
    \vspace{-1.5em}
\end{figure}


\subsection{Digital SemCom Framework}
\vspace{-0.3em}

Each semantic object is represented by a pair $(S,\Theta)$, where $S$ denotes an image and $\Theta$ denotes a latent task label associated with $S$.
The transmitter at Alice cascades a semantic encoder $E_{\varphi}$ with parameter $\varphi$ that compresses $S$ into a latent vector $\mathbf{u}\in\mathbb{R}^{2L}$, an invertible transport mapping $F_{\theta}$ with parameter $\theta$ that reshapes $\mathbf{u}$ into $\mathbf{v}\in\mathbb{R}^{2L}$, and a structured digital mapping $\mathcal{M}$ that treats the $2L$ real coordinates of $\mathbf{v}$ as $L$ complex symbols on a $K$-point constellation $\mathcal{X}$ as
\vspace{-0.4em}
\begin{equation}
\mathbf{x}=\mathcal{M}\big(F_{\theta}(E_{\varphi}(S))\big)\in\mathcal{X}^{L}.
\label{tx_chain}
\end{equation}
The distribution of the transmitted symbols is denoted by $\mu$ with $\mu(a)=\Pr\{x_{k}=a\}$ for $a\in\mathcal{X}$, with the symbols being normalized to unit average power $\mathbb{E}_{\mu}[|x_{k}|^{2}]=1$.
The scrambling and interleaving operations inside $\mathcal{M}$ are driven by a source-independent pseudo-random sequence that is refreshed per block, shared by Alice and Bob while unknown to Willie.

The end-to-end performance of SemCom is measured by distortion functions. Specifically, $\varrho_{\mathrm{o}}(s,\hat{s})$ denotes the reconstruction distortion between the source image $s$ and the reconstructed image $\hat{s}$, and $\varrho_{\mathrm{s}}(\theta,\hat{\theta})$ denotes the task distortion between the true label $\theta$ and the inferred label $\hat{\theta}$, respectively.
Since $\Theta$ is latent, the task distortion is measured through the posterior average
\vspace{-0.4em}
\begin{equation}
\bar{\varrho}_{\mathrm{s}}(s,\hat{\theta})=\mathbb{E}\big[\varrho_{\mathrm{s}}(\Theta,\hat{\theta})\,\big|\,S=s\big].
\label{post_dist}
\end{equation}
The weighted semantic distortion is defined as $d_{\beta}=\varrho_{\mathrm{o}}+\beta\bar{\varrho}_{\mathrm{s}}$, where the coefficient $\beta>0$ aims to balance the reconstruction quality and the task accuracy, and Bob employs a semantic decoder $G_{\psi}$ with parameter $\psi$ to recover the estimated pair $(\widehat{S},\widehat{\Theta})$.


\subsection{Transmission Model}
\vspace{-0.3em}

The $M$ MAs at Alice are connected to a single radio-frequency~(RF) chain via flexible cables and are repositioned by drivers within a two-dimensional transmit region $\mathcal{A}$, where the position of the $m$-th MA is denoted by $\mathbf{t}_{m}\in\mathcal{A}$.
All positions are collected into the antenna position vector (APV)
\vspace{-0.4em}
\begin{equation}
\begin{aligned}
\mathbf{T}&=[\mathbf{t}_{1},\dots,\mathbf{t}_{M}]^{\top}\in\mathcal{T},\\
\mathcal{T}&=\big\{\mathbf{T}:\mathbf{t}_{m}\in\mathcal{A},\ \|\mathbf{t}_{m}-\mathbf{t}_{m'}\|\ge d_{\min},\ \forall m\ne m'\big\},
\end{aligned}
\label{apv}
\end{equation}
where $d_{\min}$ denotes the minimum inter-MA distance that avoids mutual coupling.
The APV is reconfigured before each covert session and remains fixed over the $n$ channel uses within the session.

Under the far-field condition where the movement of an MA changes the phase of each path response without altering the path gain~\cite{zhu2024modeling}, the channel matrices from Alice to Bob and to Willie are expressed as
\vspace{-0.4em}
\begin{equation}
\mathbf{H}_{r}(\mathbf{T})=\sum_{\ell=1}^{L_{p}}\sigma_{r,\ell}\,\mathbf{c}_{r}(\varpi_{r,\ell})\,\mathbf{a}^{H}(\vartheta_{r,\ell};\mathbf{T}),\quad r\in\{b,w\},
\label{channel_bob}
\end{equation}
where $L_{p}$ denotes the number of scattering paths,
$\sigma_{r,\ell}\sim\mathcal{CN}(0,1/L_{p})$ denotes the complex
gain of the $\ell$-th path, and
$\mathbf{c}_{r}(\varpi_{r,\ell})$ and
$\mathbf{a}(\vartheta_{r,\ell};\mathbf{T})$ denote the receive
and transmit array response vectors, where $\varpi_{r,\ell}$ and $\vartheta_{r,\ell}$ denote the corresponding angle of arrival/departure~(AoA/AoD), respectively.
The transmit field-response vector is
\vspace{-0.4em}
\begin{equation}
\big[\mathbf{a}(\vartheta;\mathbf{T})\big]_{m}=e^{\mathrm{j}2\pi\mathbf{t}_{m}^{\top}\mathbf{k}(\vartheta)},
\label{frv}
\end{equation}
where $\mathbf{k}(\vartheta)=[\cos\vartheta,\ \sin\vartheta]^{\top}$ denotes the normalized wave vector.
By applying channel estimation for the legitimate link, Alice acquires the path gains and angles and therefore predicts $\mathbf{H}_{b}(\mathbf{T})$ for an arbitrary APV, whereas the distribution of the departure angles is the sole knowledge of the wiretap link at Alice.

The analog beamforming vector $\mathbf{w}\in\mathbb{C}^{M}$ behind the RF chain satisfies $\|\mathbf{w}\|=1$.
The transmitted signal at the $k$-th channel use is $\mathbf{s}_{k}=\sqrt{P}\,\mathbf{w}\,x_{k}$ with the transmit power $P$, and the received signals are
\vspace{-0.4em}
\begin{equation}
\mathbf{y}_{r,k}=\sqrt{P}\,\mathbf{H}_{r}(\mathbf{T})\,\mathbf{w}\,x_{k}+\mathbf{n}_{r,k},\quad r\in\{b,w\},
\label{tx_rx}
\end{equation}
where $\mathbf{n}_{b,k}\sim\mathcal{CN}(\mathbf{0},\sigma_{b}^{2}\mathbf{I}_{N})$ and $\mathbf{n}_{w,k}\sim\mathcal{CN}(\mathbf{0},\sigma_{w}^{2}\mathbf{I}_{M_{w}})$ denote the AWGN at Bob and Willie, respectively.
Maximum-ratio combining~(MRC) is performed at Bob and at the $M_{w}$-antenna Willie, which yields the equivalent scalar channels
\vspace{-0.4em}
\begin{equation}
\begin{aligned}
y_{b,k}&=\sqrt{\gamma_{b}}\,x_{k}+n_{b,k},\quad
\gamma_{b}=\frac{P\,\|\mathbf{H}_{b}(\mathbf{T})\mathbf{w}\|^{2}}{\sigma_{b}^{2}},\\
z_{k}&=\sqrt{\gamma_{w}}\,x_{k}+n_{w,k},\quad
\gamma_{w}=\frac{P\,\|\mathbf{H}_{w}(\mathbf{T})\mathbf{w}\|^{2}}{\sigma_{w}^{2}},
\end{aligned}
\label{scalar_ch}
\end{equation}
where $n_{b,k}$ and $n_{w,k}$ follow $\mathcal{CN}(0,1)$ while $\gamma_{b}$ and $\gamma_{w}$ denote the receive signal-to-noise ratios (SNRs) at Bob and Willie.

Since $\gamma_{w}$ involves the wiretap channel that is unknown to Alice, it is separated into a deterministic part controlled by Alice and a normalized random part whose distribution depends on the APV.
The array gain toward $\vartheta$ is defined as
\vspace{-0.4em}
\begin{equation}
\Upsilon(\vartheta;\mathbf{T})=\big|\mathbf{a}^{H}(\vartheta;\mathbf{T})\mathbf{w}\big|^{2}.
\label{array_gain}
\end{equation}
The average received energy at Willie is governed by the expectation of $\Upsilon(\vartheta;\mathbf{T})$ over the distribution of the departure angles.
Accordingly, we define
\vspace{-0.4em}
\begin{equation}
\bar{P}=\frac{P\,M_{w}\,\mathbb{E}_{\vartheta}[\Upsilon(\vartheta;\mathbf{T})]}{\sigma_{w}^{2}},\quad
U=\frac{\gamma_{w}}{\bar{P}},\quad
\Delta=\frac{\gamma_{b}}{\bar{P}},
\label{avg_snr}
\end{equation}
where $\bar{P}$ denotes the average $\gamma_{w}$ under the control of Alice, $U$ denotes the normalized fading factor with $\mathbb{E}[U]=1$ whose distribution is shaped by $\Upsilon(\vartheta;\mathbf{T})$ over the departure angles, and $\Delta$ characterizes the SNR advantage of the legitimate link over the wiretap link.
Since $U$ is random, the power control reserves an outage margin.
For a given outage tolerance $\delta\in(0,1)$, the $(1-\delta)$-quantile $u_{1-\delta}$ of $U$ satisfies
\vspace{-0.4em}
\begin{equation}
\Pr\big\{\gamma_{w}\le\bar{P}u_{1-\delta}\big\}\ge 1-\delta.
\label{outage}
\end{equation}


\subsection{Binary Hypothesis Testing at Willie}
\vspace{-0.3em}

Willie performs a binary hypothesis test to decide whether Alice is transmitting, where the receiver structure and the observation window are unknown to Alice and the noise power at Willie drifts with the receiver gain and the co-channel radiation.
Accordingly, the detection problem at Willie is formulated as the composite test
\vspace{-0.2em}
\begin{equation}
\begin{aligned}
H_{0}&:\ \mathbf{z}\sim\mathcal{CN}(\mathbf{0},s\mathbf{I}_{n}),\quad s>0,\\
H_{1}&:\ z_{k}=\sqrt{\bar{P}U}\,x_{k}+n_{w,k},\quad k=1,\dots,n,
\end{aligned}
\label{hypotheses}
\end{equation}
where $\{\mathcal{CN}(\mathbf{0},s\mathbf{I}_{n})\}_{s>0}$ constitutes the composite null hypothesis induced by the unknown noise level $s$.
Since learning-based detectors have enhanced the detection capability of Willie beyond the conventional energy detection~\cite{qian2025adversarial}, the covertness guarantee in the sequel is established against the optimal test under the composite null hypothesis rather than against a specific detector.

Moreover, the detectability is quantified by the session leakage
\vspace{-0.4em}
\begin{equation}
\mathcal{L}_{n}=\inf_{s>0}D\big(\mathbb{P}^{(n)}_{1}\,\big\|\,\mathcal{CN}(\mathbf{0},s\mathbf{I}_{n})\big),
\label{leakage}
\end{equation}
where $\mathbb{P}^{(n)}_{1}$ denotes the joint distribution of the $n$ observations under $H_{1}$.
Let $\mathcal{D}$ denote the decision region at Willie. The worst-case sum of the false-alarm and missed-detection probabilities over the noise level $s$ satisfies
\vspace{-0.4em}
\begin{equation}
\xi^{\star}=\inf_{\mathcal{D}}\sup_{s>0}\Big[\mathcal{CN}(\mathbf{0},s\mathbf{I}_{n})(\mathcal{D})+\mathbb{P}^{(n)}_{1}(\mathcal{D}^{c})\Big]\ge 1-\sqrt{\mathcal{L}_{n}/2},
\label{minimax_bound}
\end{equation}
where $\xi^{\star}$ denotes the minimax risk of Willie and the inequality follows from the minimax inequality together with the Pinsker inequality.
Given a covertness budget $\epsilon\in(0,1)$, the covertness constraint is imposed as
\vspace{-0.4em}
\begin{equation}
\mathcal{L}_{n}\le 2\epsilon^{2},
\label{covert_constraint}
\end{equation}
which guarantees the minimum risk $\xi^{\star}\ge 1-\epsilon$~\cite{bash2013limits}.
Owing to the randomness of $U$, the constraint~\eqref{covert_constraint} is required to hold with probability at least $1-\delta$.


\section{Problem Formulation and Covert Rate Optimization}\label{sec: optimization}


\subsection{Problem Formulation}\label{sec: problem}
\vspace{-0.3em}

We aim to minimize the end-to-end semantic distortion under the covertness requirement by jointly optimizing all transceiver variables introduced in Section~\ref{sec: system_model}.
The optimization is formulated as
\vspace{-0.4em}
\begin{equation}
\begin{aligned}
(\mathcal{P}1):\ \min_{\substack{\varphi,\theta,\psi,\mathcal{X},\mu\\ \mathbf{T},\mathbf{w},P}}\quad
&\mathbb{E}\big[\varrho_{\mathrm{o}}(S,\widehat{S})\big]+\beta\,\mathbb{E}\big[\bar{\varrho}_{\mathrm{s}}(S,\widehat{\Theta})\big]\\
\mathrm{s.t.}\quad
&\Pr\big\{\mathcal{L}_{n}\le 2\epsilon^{2}\big\}\ge 1-\delta,\\
&\mathbf{T}\in\mathcal{T},\quad \|\mathbf{w}\|=1,\\
&0\le P\le P_{\max},\quad \max_{a\in\mathcal{X}}|a|^{2}\le\Pi,
\end{aligned}
\label{joint_problem}
\end{equation}
where $P_{\max}$ denotes the transmit power budget and $\Pi$ denotes the peak power constraint induced by the linear region of the power amplifier.
Due to the unit average power of $\mu$, the term $\max_{a\in\mathcal{X}}|a|^{2}$ equals the peak-to-average power ratio~(PAPR).

The main difficulty for solving $(\mathcal{P}1)$ lies in the covertness constraint, which involves the joint distribution of the $n$ observations and couples the symbol distribution $\mu$ with the APV through the detectability at Willie.


\subsection{Covert Rate Characterization}\label{sec: covert_rate}
\vspace{-0.3em}

The structured mapping $\mathcal{M}$ together with the block-wise refreshed pseudo-random sequence enforces independent and identically distributed (i.i.d.) symbols with the common distribution $\mu$ over the session, under which the session leakage in~\eqref{leakage} reduces to a single-letter expression over the equivalent scalar channel in~\eqref{scalar_ch}.
Different from Gaussian signaling in existing analyses where the received distribution at Willie is specified by the received power~\cite{yan2019gaussian}, the finite constellation $\mathcal{X}$ yields a received distribution that differs in shape from $\mathcal{CN}(0,s)$ for $s>0$.

Let $Y=\sqrt{\gamma}\,X+W$ denote the per-symbol input-output relationship of the equivalent scalar channel at an arbitrary receive SNR $\gamma$ with $X\sim\mu$ and $W\sim\mathcal{CN}(0,1)$, and $\mathbb{P}_{Y}$ denote the distribution of $Y$.
Since the composite null hypothesis in~\eqref{hypotheses} is invariant to a common scaling of the observations, the normalized observation $Z=Y/\sqrt{1+\gamma}$ with $\mathbb{E}[|Z|^{2}]=1$ is considered, which correlates with $X$ as
\vspace{-0.4em}
\begin{equation}
\mathbb{E}\big[ZX^{*}\big]=\sqrt{\gamma/(1+\gamma)}=\sqrt{\alpha},
\label{alpha_def}
\end{equation}
where $\alpha\in(0,1)$ measures the signal component contained in $Z$ and approaches $\gamma$ at the low SNR region.
The following lemma expresses the per-symbol leakage in closed-form.

\begin{lemma}[\textbf{Non-Gaussianness of the Finite-Constellation Output}]\label{lemma_ng}
For the channel $Y=\sqrt{\gamma}\,X+W$ with $X\sim\mu$ and $W\sim\mathcal{CN}(0,1)$, the minimum divergence between $\mathbb{P}_{Y}$ and the composite null hypothesis is given by
\vspace{-0.4em}
\begin{equation}
\Psi_{\mu}(\gamma)=\inf_{s>0}D\big(\mathbb{P}_{Y}\,\big\|\,\mathcal{CN}(0,s)\big)=\log(1+\gamma)-I_{\mu}(\gamma),
\label{non_gaussianness}
\end{equation}
where $I_{\mu}(\gamma)=I(X;Y)$ under the input $\mu$.
$\Psi_{\mu}(\gamma)$ is strictly increasing in $\gamma\ge 0$ with $\Psi_{\mu}(0)=0$ for non-Gaussian $\mu$.
\end{lemma}

\begin{IEEEproof}
See Appendix~\ref{app:lemmas}.
\end{IEEEproof}
\vspace{-0.25em}

The quantity $\Psi_{\mu}$ is termed the output non-Gaussianness and enters both receivers.
Specifically, $\Psi_{\mu}(\gamma_{w})$ is the per-symbol leakage at Willie while $\log(1+\gamma_{b})-\Psi_{\mu}(\gamma_{b})$ is the achievable rate at Bob under the input distribution $\mu$.
A smaller $\Psi_{\mu}$ therefore permits a higher transmit power and a smaller rate loss.

Under the covertness constraint~\eqref{covert_constraint}, each channel use receives a leakage budget of $2\epsilon^{2}/n$ that vanishes as the session length $n$ grows.
As a result, the design is governed by the low-SNR behavior of $\Psi_{\mu}$ determined by the moments of $\mu$.
Let $H_{p,q}$ denote the complex Hermite polynomials that form an orthogonal basis under $\mathcal{CN}(0,1)$ and satisfy $\mathbb{E}[H_{p,q}(X')]=0$ for $X'\sim\mathcal{CN}(0,1)$ and $p+q\ge 1$.
The moment matching order of $\mu$ is defined as
\vspace{-0.4em}
\begin{equation}
t=\max\big\{\tau:\ \mathbb{E}_{\mu}[H_{p,q}(X)]=0,\ \forall\, 1\le p+q\le\tau\big\},
\label{matching_order}
\end{equation}
which indicates that all moments of $\mu$ up to order $t$ coincide with the corresponding moments of $\mathcal{CN}(0,1)$, and the lowest-order nonzero Hermite moments are collected into
\vspace{-0.4em}
\begin{equation}
\kappa_{t}=\frac{1}{2}\sum_{p+q=t+1}\frac{\big|\mathbb{E}_{\mu}[H_{p,q}(X)]\big|^{2}}{p!\,q!}.
\label{kappa_def}
\end{equation}
The following lemma characterizes the low-SNR behavior of the output non-Gaussianness.

\begin{lemma}[\textbf{Low-SNR Expansion of the Output Non-Gaussianness}]\label{lemma_expansion}
For a zero-mean constellation distribution $\mu$ with moment matching order $t$ and coefficient $\kappa_{t}$, the output non-Gaussianness admits the expansion
\vspace{-0.4em}
\begin{equation}
\Psi_{\mu}(\gamma)=\kappa_{t}\,\alpha^{t+1}\big(1+O(\alpha)\big).
\label{low_snr}
\end{equation}
\end{lemma}
\vspace{-0.25em}
\begin{IEEEproof}
See Appendix~\ref{app:lemmas}.
\end{IEEEproof}
\vspace{-0.25em}

Lemma~\ref{lemma_expansion} reveals that the per-symbol leakage scales with the $(t+1)$-th power of $\alpha$.
Thus, a larger matching order permits a higher transmit power under the same budget, and we have the following theorem.

\begin{theorem}[\textbf{Maximum Covert Rate}]\label{theorem_rate}
Suppose that $\mathbf{T}$ and $\mathbf{w}$ remain fixed within one session and that the transmitted symbols are i.i.d. with distribution $\mu$ of moment matching order $t$.
The session leakage in~\eqref{leakage} equals $\mathcal{L}_{n}=n\Psi_{\mu}(\bar{P}U)$, and the covertness constraint~\eqref{covert_constraint} holds with probability no less than $1-\delta$ if and only if
\vspace{-0.4em}
\begin{equation}
\bar{P}\le\bar{P}^{\star}=\frac{1}{u_{1-\delta}}\Psi_{\mu}^{-1}\Big(\frac{2\epsilon^{2}}{n}\Big).
\label{max_power}
\end{equation}
Under $\bar{P}=\bar{P}^{\star}$, the legitimate SNR and the covert rate are
\vspace{-0.4em}
\begin{equation}
\begin{aligned}
\gamma_{b}^{\star}&=\Delta\,\bar{P}^{\star}=\frac{\Delta}{u_{1-\delta}}\Psi_{\mu}^{-1}\Big(\frac{2\epsilon^{2}}{n}\Big),\\
R_{\mathrm{c}}&=I_{\mu}\big(\gamma_{b}^{\star}\big)\sim\frac{\Delta}{u_{1-\delta}}\Big(\frac{2\epsilon^{2}}{\kappa_{t}\,n}\Big)^{\frac{1}{t+1}},
\end{aligned}
\label{max_snr_rate}
\end{equation}
where $R_{\mathrm{c}}$ denotes the covert rate per channel use and the asymptotic equivalence holds as $n\to\infty$.
\end{theorem}

\begin{IEEEproof}
See Appendix~\ref{app:theorem_rate}.
\end{IEEEproof}
\vspace{-0.25em}

Theorem~\ref{theorem_rate} transforms the intractable covertness constraint in $(\mathcal{P}1)$ into the explicit power limit in~\eqref{max_power}, and the resulting rate expression is given by~\eqref{max_snr_rate}.
It is seen that the covert rate decays as $n^{-1/(t+1)}$ while the total delivered information $nR_{\mathrm{c}}$ grows as $n^{t/(t+1)}$.
For $t=1$, the growth reduces to the classic square-root law in covert communication~\cite{bash2013limits}.
Moreover, $\gamma_{b}^{\star}$ factorizes into the constellation term $\Psi_{\mu}^{-1}(2\epsilon^{2}/n)$ and the antenna term $\Delta/u_{1-\delta}$, which allows the constellation and the APV to be designed separately, as shown in the following subsections.


\begin{figure}[t]
    \centering
    \subfigure[Constellations with matching orders.]{%
        \includegraphics[width=.5\linewidth]{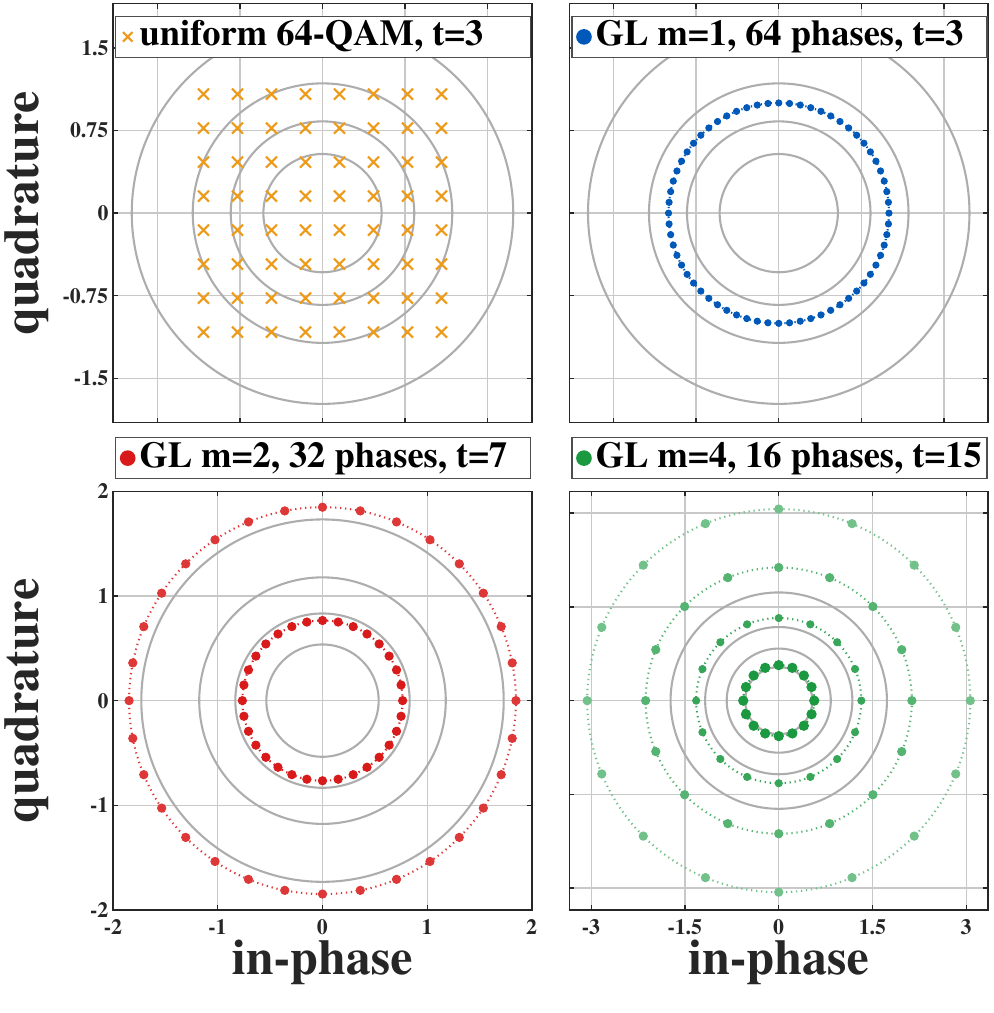}%
        \label{constellation}%
    }%
    \hfill%
    \subfigure[MA positions in two sessions with $d_{\min}=\lambda/2$.]{%
        \includegraphics[width=.5\linewidth]{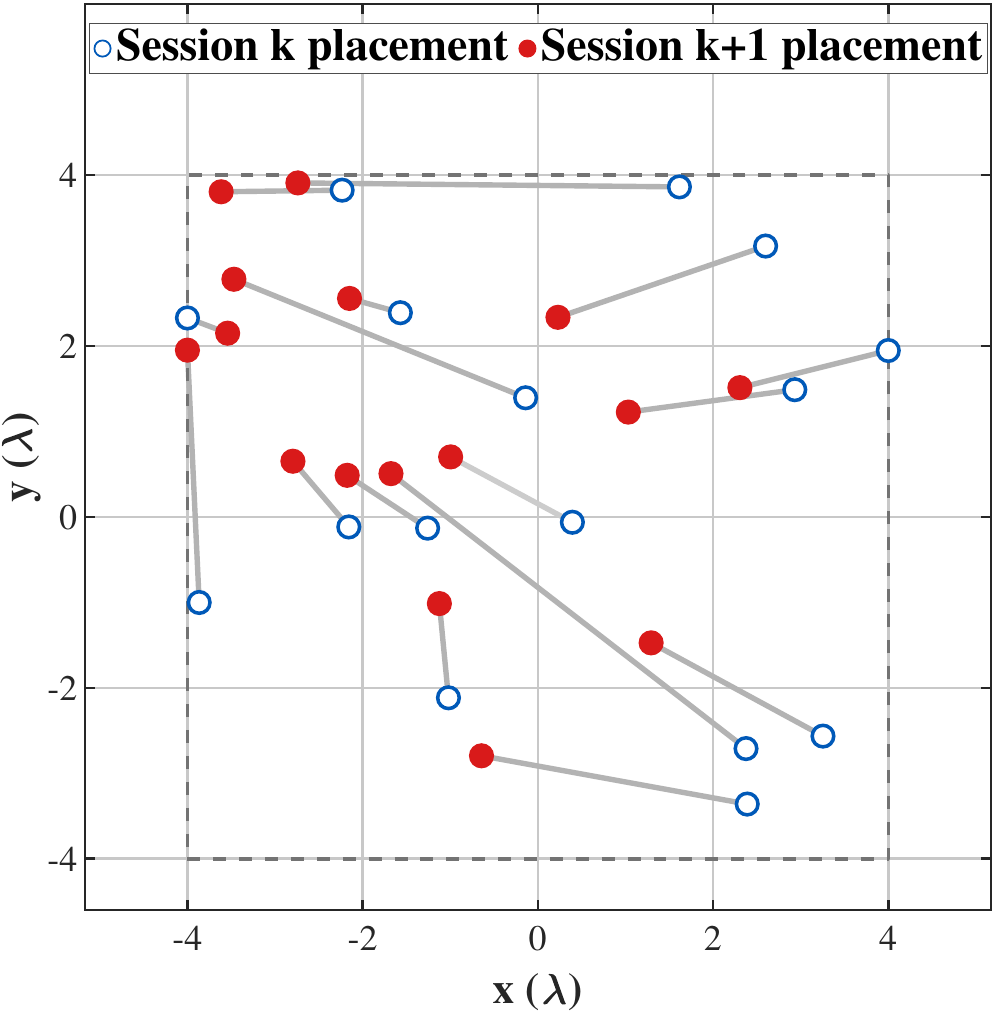}%
        \label{array_sessions}%
    }%
    \vspace{-0.5em}
    \caption{Designed constellations on the quantile circles of $\mathcal{CN}(0,1)$ and the MA positions of two consecutive sessions within the region $\mathcal{A}$.}
    \label{fig:const_array}
    \vspace{-1.5em}
\end{figure}

\subsection{Gaussian-Matched Constellation Design}\label{sec: constellation_design}
\vspace{-0.3em}

To maximize the constellation term $\Psi_{\mu}^{-1}(2\epsilon^{2}/n)$ in~\eqref{max_snr_rate} under the PAPR constraint in $(\mathcal{P}1)$, the matching order $t$ is raised as far as the PAPR constraint permits, since $\Psi_{\mu}^{-1}(2\epsilon^{2}/n)$ grows with $t$ by~\eqref{low_snr} for the vanishing budget $2\epsilon^{2}/n$.
Since the composite null hypothesis in~\eqref{hypotheses} is invariant to a common rotation of the observations while nonzero moments $\mathbb{E}_{\mu}[x^{p}(x^{*})^{q}]$ with $p\ne q$ are exposed by the phase correlation of the received samples regardless of $s$, the symbol distribution is designed with discrete rotational symmetry under which the phase and the amplitude are matched separately.

To this end, we design a ring constellation with $m$ rings and $M_{\mathrm{ph}}$ equispaced phases per ring as
\vspace{-0.4em}
\begin{equation}
\mathcal{X}_{m}=\big\{\sqrt{\rho_{k}}\,e^{\mathrm{j}2\pi\nu/M_{\mathrm{ph}}}\big\}_{k\le m,\ \nu<M_{\mathrm{ph}}},\quad K=mM_{\mathrm{ph}},
\label{ring_const}
\end{equation}
where $\rho_{k}$ denotes the squared radius of the $k$-th ring and the points on the $k$-th ring are assigned equal probabilities, with $\omega_{k}$ denoting the probability of the $k$-th ring.
The design variables are thereby reduced from the $2K$ coordinates of $\mathcal{X}$ to the $m$ pairs $(\rho_{k},\omega_{k})$, which satisfy the unit average power and the PAPR constraints in $(\mathcal{P}1)$ as
\vspace{-0.4em}
\begin{equation}
\sum_{k=1}^{m}\omega_{k}\rho_{k}=1,\quad
\max_{a\in\mathcal{X}_{m}}|a|^{2}=\max_{k}\rho_{k}\le\Pi.
\label{ring_papr}
\end{equation}
The phase sum over the equispaced phases equals the geometric series $\frac{1}{M_{\mathrm{ph}}}\sum_{\nu=0}^{M_{\mathrm{ph}}-1}e^{\mathrm{j}2\pi\nu(p-q)/M_{\mathrm{ph}}}=\mathbf{1}\{M_{\mathrm{ph}}\mid(p-q)\}$, which factorizes every moment of $\mathcal{X}_{m}$ as
\vspace{-0.4em}
\begin{equation}
\mathbb{E}_{\mu}\big[x^{p}(x^{*})^{q}\big]
=\Big(\sum_{k=1}^{m}\omega_{k}\rho_{k}^{\frac{p+q}{2}}\Big)\mathbf{1}\Big\{\frac{p-q}{M_{\mathrm{ph}}}\in\mathbb{Z}\Big\}.
\label{ring_moment}
\end{equation}
All off-diagonal moments with $0<|p-q|<M_{\mathrm{ph}}$ vanish as in $\mathcal{CN}(0,1)$.
Then, the matching condition reduces to
\vspace{-0.4em}
\begin{equation}
\sum_{k=1}^{m}\omega_{k}\rho_{k}^{\,j}=\int_{0}^{\infty}e^{-\rho}\rho^{j}\,\mathrm{d}\rho=j!,\quad j=1,\dots,\Big\lfloor\frac{t}{2}\Big\rfloor,
\label{ring_quad}
\end{equation}
which contains the unit power constraint as the case $j=1$ and is an $m$-node quadrature rule for the weight function $e^{-\rho}$, where $(\rho_{k},\omega_{k})$ are determined in the following proposition to attain the highest matching order.

\begin{proposition}[\textbf{Maximum Matching Order of $m$-Ring Constellations}]\label{prop_gl}
For the ring constellation in~\eqref{ring_const} with $M_{\mathrm{ph}}\ge 4m$ under the PAPR constraint in~\eqref{ring_papr}, the moment matching order satisfies $t\le 4m-1$, and the bound is attained when $(\rho_{k},\omega_{k})$ are the nodes and weights of the $m$-point GL quadrature for the weight function $e^{-\rho}$ and the number of rings is the largest $m$ with $\max_{k}\rho_{k}\le\Pi$.
\end{proposition}

\begin{IEEEproof}
Since $H_{p,q}$ equals $x^{p}(x^{*})^{q}$ plus terms of lower total degree, the Hermite moments in~\eqref{matching_order} vanish for $1\le p+q\le\tau$ if and only if $\mathbb{E}_{\mu}[x^{p}(x^{*})^{q}]=p!\,\mathbf{1}\{p=q\}$ holds for $1\le p+q\le\tau$.
Under $M_{\mathrm{ph}}\ge 4m$ and $p+q\le 4m-1$, the condition reduces to~\eqref{ring_quad}.
The $m$-node rule in~\eqref{ring_quad} is inexact for the nonnegative polynomial $f(\rho)=\prod_{k=1}^{m}(\rho-\rho_{k})^{2}$ of degree $2m$ since $\sum_{k=1}^{m}\omega_{k}f(\rho_{k})=0<\int_{0}^{\infty}e^{-\rho}f(\rho)\,\mathrm{d}\rho$, therefore $\lfloor t/2\rfloor\le 2m-1$ and $t\le 4m-1$.
Conversely, the $m$-point GL quadrature is exact for all polynomials of degree up to $2m-1$ with positive nodes and weights~\cite{gautschi2004orthogonal}, which satisfies~\eqref{ring_quad} for $j\le 2m-1$ together with the unit power constraint and attains the bound.
The outermost node $\max_{k}\rho_{k}$ grows with $m$, which bounds the admissible number of rings through~\eqref{ring_papr}.
\end{IEEEproof}
\vspace{-0.25em}

Proposition~\ref{prop_gl} reduces the constellation design to the choice of $(m,M_{\mathrm{ph}})$, and the resulting GL constellation $\mathcal{X}_{\mathrm{GL}}$ with distribution $\mu_{\mathrm{GL}}$ attains $t=4m-1$, which enters the APV design in the next subsection.
As depicted in Fig.~\ref{constellation}, the rings of $\mathcal{X}_{\mathrm{GL}}$ spread outward with $m$ along the quantile circles of $\mathcal{CN}(0,1)$ and raise the matching order from $t=3$ of uniform 64-QAM to $t=15$ at $m=4$.


\subsection{MA Placement and Power Control}\label{sec: apv_power}
\vspace{-0.3em}

The APV design aims to maximize the antenna term $\Delta/u_{1-\delta}$ in~\eqref{max_snr_rate}, which requires the array to enlarge the channel gain $\|\mathbf{H}_{b}(\mathbf{T})\mathbf{w}\|^{2}$ of the legitimate link and to suppress the transmit array gain $\Upsilon(\vartheta;\mathbf{T})$ over the entire angular range since the direction of Willie is unknown.
The angular average in~\eqref{avg_snr} can be characterized by the transmit correlation matrix
\vspace{-0.4em}
\begin{equation}
\begin{aligned}
\bar{\mathbf{R}}_{w}(\mathbf{T})&=\mathbb{E}_{\vartheta}\big[\mathbf{a}(\vartheta;\mathbf{T})\mathbf{a}^{H}(\vartheta;\mathbf{T})\big],\\
\big[\bar{\mathbf{R}}_{w}(\mathbf{T})\big]_{mm'}&=J_{0}\big(2\pi\|\mathbf{t}_{m}-\mathbf{t}_{m'}\|\big),
\end{aligned}
\label{Rw_def}
\end{equation}
which yields $\mathbb{E}_{\vartheta}[\Upsilon(\vartheta;\mathbf{T})]=\mathbf{w}^{H}\bar{\mathbf{R}}_{w}(\mathbf{T})\mathbf{w}$, where the entries of $\bar{\mathbf{R}}_{w}(\mathbf{T})$ depend on the inter-MA distances through the zeroth-order Bessel function $J_{0}(\cdot)$ under isotropic scattering around Willie.
Let $\mathbf{w}_{b}(\mathbf{T})=\arg\max_{\|\mathbf{w}\|=1}\|\mathbf{H}_{b}(\mathbf{T})\mathbf{w}\|^{2}$ denote the unit-norm dominant eigenvector of $\mathbf{H}_{b}^{H}(\mathbf{T})\mathbf{H}_{b}(\mathbf{T})$, and the following corollary specifies the APV that maximizes $\gamma_{b}^{\star}$.

\begin{corollary}[\textbf{Matching-Order-Aware MA Placement}]\label{cor_place}
Under $\mathbf{w}=\mathbf{w}_{b}(\mathbf{T})$, the APV that maximizes $\gamma_{b}^{\star}$ in~\eqref{max_snr_rate} is
\vspace{-0.4em}
\begin{equation}
\mathbf{T}^{\star}\in\arg\max_{\mathbf{T}\in\mathcal{T}}\
\frac{\lambda_{\max}\big(\mathbf{H}_{b}^{H}(\mathbf{T})\mathbf{H}_{b}(\mathbf{T})\big)}
{M_{w}\,\mathbf{w}_{b}^{H}\bar{\mathbf{R}}_{w}(\mathbf{T})\mathbf{w}_{b}\,u_{1-\delta}(\mathbf{T})}.
\label{place_exact}
\end{equation}
Moreover, when the covertness constraint is imposed on the average leakage as $\mathbb{E}_{U}[\mathcal{L}_{n}]\le 2\epsilon^{2}$, the admissible average warden SNR satisfies
\vspace{-0.4em}
\begin{equation}
\bar{P}^{\star}_{\mathrm{av}}\sim\frac{1}{\varsigma_{t}(\mathbf{T})}\Psi_{\mu}^{-1}\Big(\frac{2\epsilon^{2}}{n}\Big),\quad
\varsigma_{t}(\mathbf{T})=\big(\mathbb{E}[U^{t+1}]\big)^{\frac{1}{t+1}},
\label{fading_norm}
\end{equation}
as $n\to\infty$, where the fading norm $\varsigma_{t}(\mathbf{T})$ is the $(t+1)$-th order norm of $U$ and is non-decreasing with respect to $t$.
\end{corollary}

\begin{IEEEproof}
Under $\mathbf{w}=\mathbf{w}_{b}(\mathbf{T})$, the maximum beamforming gain equals $\lambda_{\max}(\mathbf{H}_{b}^{H}(\mathbf{T})\mathbf{H}_{b}(\mathbf{T}))$, and~\eqref{scalar_ch} with~\eqref{Rw_def} in~\eqref{avg_snr} yield
\vspace{-0.4em}
\begin{equation}
\frac{\Delta(\mathbf{T})}{u_{1-\delta}(\mathbf{T})}
=\frac{\sigma_{w}^{2}}{\sigma_{b}^{2}}\cdot
\frac{\lambda_{\max}\big(\mathbf{H}_{b}^{H}(\mathbf{T})\mathbf{H}_{b}(\mathbf{T})\big)}
{M_{w}\,\mathbf{w}_{b}^{H}\bar{\mathbf{R}}_{w}(\mathbf{T})\mathbf{w}_{b}\,u_{1-\delta}(\mathbf{T})}.
\label{app_ratio}
\end{equation}
Since the noise power ratio $\frac{\sigma_{w}^{2}}{\sigma_{b}^{2}}$ is independent of $\mathbf{T}$, maximizing $\gamma_{b}^{\star}$ in~\eqref{max_snr_rate} over $\mathcal{T}$ is equivalent to solving~\eqref{place_exact}.
For the second part, Theorem~\ref{theorem_rate} gives $\mathcal{L}_{n}=n\Psi_{\mu}(\bar{P}U)$ for every realization of $U$, and~\eqref{low_snr} with $\alpha=\gamma(1+O(\gamma))$ yields $\mathbb{E}_{U}[\mathcal{L}_{n}]=n\kappa_{t}\bar{P}^{t+1}\mathbb{E}[U^{t+1}](1+O(\bar{P}))$ since $U$ has finite moments of all orders.
Setting the average leakage equal to $2\epsilon^{2}$ and solving for $\bar{P}$ yields
\vspace{-0.4em}
\begin{equation}
\begin{aligned}
\bar{P}^{\star}_{\mathrm{av}}
&=\frac{1}{\varsigma_{t}(\mathbf{T})}\Big(\frac{2\epsilon^{2}}{\kappa_{t}\,n}\Big)^{\frac{1}{t+1}}\Big(1+O\big(n^{-\frac{1}{t+1}}\big)\Big),
\end{aligned}
\label{app_pav}
\end{equation}
which reduces to~\eqref{fading_norm} through the inverse expansion~\eqref{app_inverse} in Appendix~\ref{app:theorem_rate}.
The monotonicity of $\varsigma_{t}(\mathbf{T})$ with respect to $t$ follows from the Lyapunov inequality $(\mathbb{E}[U^{r}])^{1/r}\le(\mathbb{E}[U^{r'}])^{1/r'}$. 
\end{IEEEproof}
\vspace{-0.25em}

Equation~\eqref{fading_norm} shows that the admissible warden SNR is divided by the $(t+1)$-th order norm of $U$, which grows with the matching order and thereby couples the APV design with the constellation design.
Since the quantile $u_{1-\delta}(\mathbf{T})$ in~\eqref{place_exact} is not differentiable with respect to $\mathbf{T}$, the placement adopts the differentiable fading norm $\varsigma_{t}(\mathbf{T})$ in~\eqref{fading_norm} that reflects the upper tail of $U$ as $u_{1-\delta}(\mathbf{T})$ does, which gives
\vspace{-0.4em}
\begin{equation}
\max_{\mathbf{T}\in\mathcal{T}}\ \log\lambda_{\max}\big(\mathbf{H}_{b}^{H}\mathbf{H}_{b}\big)
-\log\big(\mathbf{w}_{b}^{H}\bar{\mathbf{R}}_{w}\mathbf{w}_{b}\big)-\log\varsigma_{t}(\mathbf{T}).
\label{place_smooth}
\end{equation}
The derivative of the largest eigenvalue equals the derivative of the quadratic form $\mathbf{w}_{b}^{H}\mathbf{H}_{b}^{H}\mathbf{H}_{b}\mathbf{w}_{b}$ with $\mathbf{w}_{b}$ held fixed by the first-order eigenvalue perturbation, and since only the $m$-th column of $\mathbf{H}_{b}$ depends on $\mathbf{t}_{m}$, we have
\vspace{-0.4em}
\begin{equation}
\begin{aligned}
\frac{\partial\log\lambda_{\max}(\mathbf{H}_{b}^{H}\mathbf{H}_{b})}{\partial\mathbf{t}_{m}}
&=\frac{1}{\lambda_{\max}}\frac{\partial\big(\mathbf{w}_{b}^{H}\mathbf{H}_{b}^{H}\mathbf{H}_{b}\mathbf{w}_{b}\big)}{\partial\mathbf{t}_{m}}\\
&=\frac{2}{\lambda_{\max}}\mathrm{Re}\Big\{w_{b,m}\,(\mathbf{H}_{b}\mathbf{w}_{b})^{H}\frac{\partial[\mathbf{H}_{b}]_{:,m}}{\partial\mathbf{t}_{m}}\Big\},
\end{aligned}
\label{grad_lambda}
\end{equation}
where $w_{b,m}$ denotes the $m$-th entry of $\mathbf{w}_{b}$ and the column derivative follows from~\eqref{channel_bob} and~\eqref{frv} as
\vspace{-0.4em}
\begin{equation}
\frac{\partial[\mathbf{H}_{b}]_{:,m}}{\partial\mathbf{t}_{m}}
=-\mathrm{j}2\pi\sum_{\ell=1}^{L_{p}}\sigma_{b,\ell}\,\mathbf{c}_{b}(\varpi_{b,\ell})\, e^{-\mathrm{j}2\pi\mathbf{t}_{m}^{\top}\mathbf{k}(\vartheta_{b,\ell})}\,\mathbf{k}^{\top}(\vartheta_{b,\ell}).
\label{grad_column}
\end{equation}
The second and third terms in~\eqref{place_smooth} are calculated over a fixed grid of departure angles with derivatives from~\eqref{frv} by the chain rule, and the constraint $\mathbf{T}\in\mathcal{T}$ is enforced by projection onto $\mathcal{T}$ after each gradient step with a backtracking step size, which yields a nondecreasing objective sequence bounded on the compact set $\mathcal{T}$ and therefore converges to a stationary point of~\eqref{place_smooth}.

Fig.~\ref{array_sessions} illustrates the optimized MA positions for two consecutive sessions, where each element is relocated within $\mathcal{A}$ as the legitimate channel changes.
The transmit power is finally set by the exact quantile $u_{1-\delta}(\mathbf{T}^{\star})$ through~\eqref{max_power}, which keeps the covertness guarantee of Theorem~\ref{theorem_rate} intact after the replacement of $u_{1-\delta}$ by $\varsigma_{t}$ in~\eqref{place_smooth}.

With $\mathbf{w}^{\star}=\mathbf{w}_{b}(\mathbf{T}^{\star})$, the relations in~\eqref{avg_snr} and~\eqref{Rw_def} yield the transmit power and the legitimate SNR
\vspace{-0.4em}
\begin{equation}
\begin{aligned}
P^{\star}&=\min\Big\{P_{\max},\
\frac{\sigma_{w}^{2}\,\bar{P}^{\star}(\mu_{\mathrm{GL}},\mathbf{T}^{\star})}
{M_{w}\,\mathbf{w}^{\star H}\bar{\mathbf{R}}_{w}(\mathbf{T}^{\star})\mathbf{w}^{\star}}\Big\},\\
\gamma_{b}^{\star}&=\frac{P^{\star}\lambda_{\max}\big(\mathbf{H}_{b}^{H}(\mathbf{T}^{\star})\mathbf{H}_{b}(\mathbf{T}^{\star})\big)}{\sigma_{b}^{2}},
\end{aligned}
\label{P_star}
\end{equation}
which reduces to $\gamma_{b}^{\star}=\Delta\bar{P}^{\star}$ in~\eqref{max_snr_rate} for $P^{\star}<P_{\max}$.
Consequently, the covertness constraint in $(\mathcal{P}1)$ has been converted into the fixed physical configuration $\Omega^{\star}=(\mathcal{X}_{\mathrm{GL}},\mu_{\mathrm{GL}},t,\mathbf{T}^{\star},\mathbf{w}^{\star},P^{\star},\gamma_{b}^{\star})$, over which the semantic transceiver is designed in the next section.


\section{Digital Covert SemCom Transceiver Design}\label{sec: transceiver}

\begin{figure*}[htb]
    \centering
    \includegraphics[width=\textwidth]{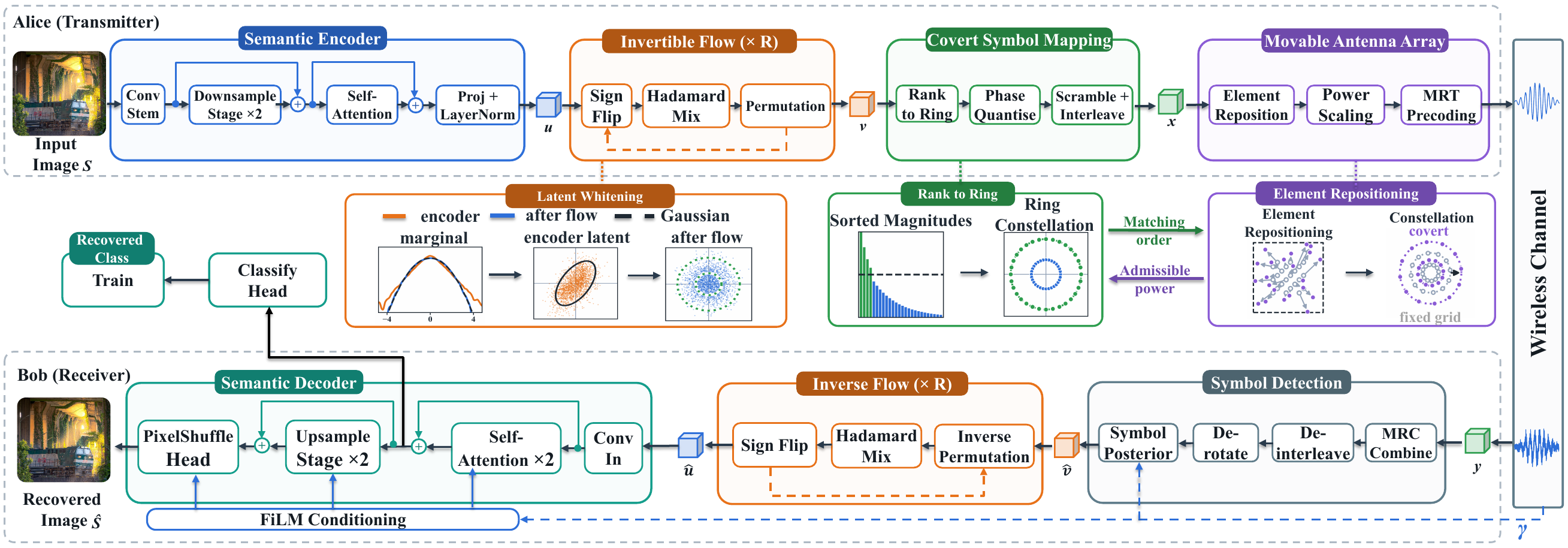}
    \vspace{-0.5em}
    \caption{Network architecture of the proposed semantic transceiver with structured digital mapping.
}
    \label{network}
    \vspace{-1.5em}
\end{figure*}


\subsection{Semantic Distortion Bound}\label{sec: distortion_bound}
\vspace{-0.3em}

To characterize the semantic performance attainable under $\Omega^{\star}$, the rate-distortion function of the weighted semantic distortion $d_{\beta}$ in Section~\ref{sec: system_model} is defined as
\vspace{-0.4em}
\begin{equation}
\mathcal{D}_{\beta}(R)=\inf_{I(S;\widehat{S},\widehat{\Theta})\le R}\mathbb{E}\big[d_{\beta}\big],
\label{rd_function}
\end{equation}
which is non-increasing and convex in the mutual information budget $R$.
The following theorem establishes a lower bound on the session-averaged distortion under the covertness constraint.

\begin{theorem}[\textbf{Semantic Distortion Bound Under Covertness}]\label{theorem_distortion}
Consider a session of $J=n/L$ independent semantic objects with arbitrary joint coding across the objects under a constellation $\mathcal{X}$, subject to the covertness constraint~\eqref{covert_constraint} with probability no less than $1-\delta$.
Let $\gamma_{b}^{\star}(\mu)$ denote the legitimate SNR in~\eqref{max_snr_rate} under $\mu$, and define the covert semantic rate
\vspace{-0.4em}
\begin{equation}
\eta_{n}^{\star}=\sup_{\mathcal{X},\mu}\ I_{\mu}\big(\gamma_{b}^{\star}(\mu)\big),
\label{covert_sem_rate}
\end{equation}
where the supremum is taken over all constellations with $|\mathcal{X}|=K$, $\max_{a\in\mathcal{X}}|a|^{2}\le\Pi$, and the distributions on them.
The session-averaged distortion $\bar{D}_{J}=\frac{1}{J}\sum_{j=1}^{J}\mathbb{E}[d_{\beta,j}]$ satisfies
\vspace{-0.4em}
\begin{equation}
\bar{D}_{J}\ge\mathcal{D}_{\beta}\big(L\,\eta_{n}^{\star}\big),
\label{distortion_bound}
\end{equation}
where $d_{\beta,j}$ denotes the weighted distortion of the $j$-th object.
\end{theorem}

\begin{IEEEproof}
See Appendix~\ref{app:theorem_distortion}.
\end{IEEEproof}
\vspace{-0.25em}

Theorem~\ref{theorem_distortion} reduces the semantic component of $(\mathcal{P}1)$ to the scalar $L\eta_{n}^{\star}$ governed by $\Omega^{\star}$, and the equality conditions associated with the leakage bound in the proof require the $L$ transmitted symbols per block to be i.i.d. with distribution $\mu$ for any source realization $S=s$ as
\vspace{-0.4em}
\begin{equation}
\mathbb{P}_{x_{1},\dots,x_{L}\mid S=s}=\prod_{k=1}^{L}\mu(x_{k}),\quad\forall s.
\label{target_law}
\end{equation}
Since~\eqref{target_law} eliminates the dependence on $s$ from the block distribution at Willie, semantic information is instead conveyed by transmitted symbols conditioned on the scrambling sequence available to Bob.

Accordingly, the transmitter includes two functional layers, i.e., a structure layer that enforces~\eqref{target_law} by source-independent randomization, and a learning layer that comprises the invertible flow and the semantic encoder and decoder to minimize the weighted distortion over the fixed configuration.


\subsection{Covert Symbol Mapping}\label{sec: structure_layer}

To enforce~\eqref{target_law}, the structure layer is implemented as the covert symbol mapping $\mathcal{M}$ in Fig.~\ref{network}, which maps $\mathbf{v}$ to the transmitted block $\mathbf{x}$ through rank-to-ring assignment, phase quantization, and scrambling with interleaving.

First, the rank-to-ring module assigns the $L$ complex coordinates of $\mathbf{v}$ to the $m$ rings of $\mathcal{X}_{\mathrm{GL}}$ under the occupancy vector $\mathbf{N}=(N_{1},\dots,N_{m})$ that specifies the number of coordinates carried by each ring.
The occupancy is generated from the pseudo-random sequence according to the conditional binomial rule
\vspace{-0.4em}
\begin{equation}
N_{i}\,\big|\,N_{1},\dots,N_{i-1}\sim\mathrm{Bin}\Big(L-\sum_{j<i}N_{j},\ \frac{\omega_{i}}{1-\sum_{j<i}\omega_{j}}\Big),
\label{occupancy}
\end{equation}
for $i=1,\dots,m-1$ with $N_{m}=L-\sum_{j<m}N_{j}$, where $\mathrm{Bin}(\cdot,\cdot)$ denotes the binomial distribution.
Multiplying the conditional probabilities in~\eqref{occupancy} cancels the successive factorials and yields the multinomial distribution
\vspace{-0.4em}
\begin{equation}
\Pr\{\mathbf{N}=\mathbf{n}\}=\binom{L}{\mathbf{n}}\prod_{i=1}^{m}\omega_{i}^{n_{i}},\quad
\binom{L}{\mathbf{n}}=\frac{L!}{\prod_{i=1}^{m}n_{i}!}.
\label{multinomial}
\end{equation}
Let $r_{i}\in\{1,\dots,L\}$ denote the rank of $|v_{i}|$ among the $L$ block magnitudes sorted in the ascending order as illustrated in Fig.~\ref{network}.
The ring index of the $i$-th coordinate is given by
\vspace{-0.4em}
\begin{equation}
k_{i}=\min\Big\{k:\ r_{i}\le\sum_{j\le k}N_{j}\Big\},
\label{rank_to_ring}
\end{equation}
which allocates the $N_{1}$ coordinates with the smallest magnitude to the ring of the smallest radius and depends on $\mathbf{v}$ through the ordering of the magnitudes.
In parallel, the phase quantization module maps the phase of each coordinate to one of the $M_{\mathrm{ph}}$ equispaced phases according to $\nu_{i}=\lfloor\frac{M_{\mathrm{ph}}}{2\pi}\arg v_{i}\rfloor\bmod M_{\mathrm{ph}}$.

The scrambling and interleaving module then shifts each phase index by an independent offset $\Theta_{i}\sim\mathrm{Unif}\{0,\dots,M_{\mathrm{ph}}-1\}$ and permutes the $L$ symbols by $\Pi\sim\mathrm{Unif}(\mathfrak{S}_{L})$ with $\mathfrak{S}_{L}$ the symmetric group of degree $L$, which removes the dependence on the source and yields
\vspace{-0.4em}
\begin{equation}
\mathbf{x}=\Pi\Big(\Big\{\sqrt{\rho_{k_{i}}}\,e^{\mathrm{j}2\pi(\nu_{i}+\Theta_{i})/M_{\mathrm{ph}}}\Big\}_{i=1}^{L}\Big).
\label{tx_block}
\end{equation}
Conditioned on $\mathbf{N}=\mathbf{n}$, the multiset of the $L$ magnitudes is determined by $\mathbf{n}$, and the uniform permutation with the independent uniform offsets equalizes positions and phases as
\vspace{-0.4em}
\begin{equation}
\mathbb{P}_{\mathbf{x}\mid\mathbf{N}=\mathbf{n},S=s}(\mathbf{a})
=\frac{\mathbf{1}\{\mathbf{a}\in\mathcal{B}(\mathbf{n})\}}{\binom{L}{\mathbf{n}}M_{\mathrm{ph}}^{L}},\quad\forall s,
\label{cond_block}
\end{equation}
where $\mathcal{B}(\mathbf{n})$ collects the blocks whose ring occupancy equals $\mathbf{n}$.
Averaging~\eqref{cond_block} over the law of $\mathbf{N}$ in~\eqref{multinomial} and counting the $\binom{L}{\mathbf{n}}M_{\mathrm{ph}}^{L}$ blocks of $\mathcal{B}(\mathbf{n})$ lead to
\vspace{-0.4em}
\begin{equation}
\mathbb{P}_{\mathbf{x}\mid S=s}(\mathbf{a})
=\frac{1}{M_{\mathrm{ph}}^{L}}\prod_{i=1}^{m}\omega_{i}^{n_{i}(\mathbf{a})}
=\prod_{k=1}^{L}\mu_{\mathrm{GL}}(a_{k}),
\label{product_law}
\end{equation}
where $n_{i}(\mathbf{a})$ denotes the number of symbols of $\mathbf{a}$ on the $i$-th ring and the second equality follows from $\mu_{\mathrm{GL}}(a)=\omega_{k(a)}/M_{\mathrm{ph}}$ in~\eqref{ring_const}.
Equation~\eqref{product_law} establishes~\eqref{target_law} with $\mu=\mu_{\mathrm{GL}}$, under which the block distribution at Willie is fixed by construction and the covert rate of Theorem~\ref{theorem_rate} is attained.


\subsection{Digital Semantic Encoder and Invertible Flow}\label{sec: learning_layer}
\vspace{-0.3em}

To maximize the semantic information carried by the magnitude ordering and the quantized phases in Section~\ref{sec: structure_layer}, the digital semantic encoder is followed by an invertible flow, as shown in Fig.~\ref{network}.
The input image $S\in\mathbb{R}^{H\times W\times 3}$ first passes through the convolutional stem and the downsample stages, which reduce the spatial resolution by a factor of $2^{n_{s}}$.
The projection maps the resulting feature map to $C=2L\cdot 4^{n_{s}}/(HW)$ channels and flattens it into $\mathbf{f}\in\mathbb{R}^{2L}$.
Subsequently, the layer normalization standardizes $\mathbf{f}$ per image as
\vspace{-0.4em}
\begin{equation}
u_{i}=\frac{f_{i}-\bar{f}}{\sigma_{f}},\quad
\bar{f}=\frac{1}{2L}\sum_{i=1}^{2L}f_{i},\quad
\sigma_{f}^{2}=\frac{1}{2L}\sum_{i=1}^{2L}\big(f_{i}-\bar{f}\big)^{2},
\label{layer_norm}
\end{equation}
which yields the latent vector $\mathbf{u}=E_{\varphi}(S)$ with zero mean and unit variance across $2L$ entries.

To equalize the energy of the latent coordinates prior to rank-based mapping, the invertible flow is constructed via $R$ cascaded stages of sign flipping, Hadamard mixing, and permutation, given by
\vspace{-0.4em}
\begin{equation}
\mathbf{v}=F_{\theta}(\mathbf{u})=\mathbf{W}_{R-1}\cdots\mathbf{W}_{0}\,\mathbf{u},\quad
\mathbf{W}_{r}=\mathbf{P}_{r}\mathbf{H}_{2L}\mathbf{D}_{r},
\label{flow_def}
\end{equation}
where $\mathbf{D}_{r}$ is a diagonal matrix with pseudo-random entries in $\{-1,+1\}$, $\mathbf{H}_{2L}$ is the normalized Walsh--Hadamard matrix of order $2L$ with entries $\pm(2L)^{-1/2}$, and $\mathbf{P}_{r}$ is a pseudo-random permutation matrix.
The parameter vector $\theta$ subsumes the sign flips and the permutations that are fixed without learning, and the orthogonality of each factor in~\eqref{flow_def} leads to
\vspace{-0.4em}
\begin{equation}
\|\mathbf{v}\|=\|\mathbf{u}\|,\quad
F_{\theta}^{-1}(\mathbf{v})=\mathbf{W}_{0}^{\top}\cdots\mathbf{W}_{R-1}^{\top}\,\mathbf{v}.
\label{flow_inverse}
\end{equation}
For the input $\mathbf{z}$ of the $r$-th stage with independent signs in $\mathbf{D}_{r}$, the cross terms vanish in the second moment of each mixed coordinate and $\mathbb{E}_{\mathbf{D}_{r}}[(\mathbf{H}_{2L}\mathbf{D}_{r}\mathbf{z})_{i}^{2}]=\|\mathbf{z}\|^{2}/(2L)$ holds for $i=1,\dots,2L$. 
Besides, the Hoeffding inequality with a union bound limits the deviation from the mean energy as
\vspace{-0.4em}
\begin{equation}
\Pr\Big\{\|\mathbf{H}_{2L}\mathbf{D}_{r}\mathbf{z}\|_{\infty}\le\|\mathbf{z}\|\sqrt{\frac{\log(4L/p_{0})}{L}}\Big\}\ge 1-p_{0}
\label{spreading}
\end{equation}
for $p_{0}\in(0,1)$~\cite{ailon2009fast}, where $\|\cdot\|_{\infty}$ denotes the largest magnitude among the entries.
Thus, the energy of $\mathbf{u}$ is distributed uniformly over the coordinates of $\mathbf{v}$, and the elliptical scatter of the encoder latent becomes isotropic after the flow, as illustrated in the latent whitening panel of Fig.~\ref{network}.

Finally, $\mathbf{x}$ is transmitted through the MA array under the configuration $\Omega^{\star}$ with the maximum-ratio transmission~(MRT) precoding $\mathbf{w}^{\star}=\mathbf{w}_{b}(\mathbf{T}^{\star})$ and the power $P^{\star}$ in~\eqref{tx_rx}, where the arrows in the element repositioning panel of Fig.~\ref{network} indicate that the matching order $t$ affects the placement objective~\eqref{place_smooth} through $\varsigma_{t}(\mathbf{T})$ and that $\mathbf{T}^{\star}$ in turn determines the admissible power through $u_{1-\delta}(\mathbf{T}^{\star})$ in~\eqref{max_power}.


\subsection{Receiver Design and End-to-End Training}\label{sec: training}
\vspace{-0.3em}

At Bob, the received samples are converted back into the latent vector in the reverse order of the transmitter in Fig.~\ref{network}.
After the MRC in~\eqref{scalar_ch}, the de-interleave and de-rotate modules apply the inverse permutation and the phase offsets of the sequence, which recovers the $i$-th symbol as
\vspace{-0.4em}
\begin{equation}
\tilde{y}_{i}
=e^{-\mathrm{j}2\pi\Theta_{i}/M_{\mathrm{ph}}}\big[\Pi^{-1}(\mathbf{y}_{b})\big]_{i}
=\sqrt{\gamma_{b}\rho_{k_{i}}}\,e^{\mathrm{j}2\pi\nu_{i}/M_{\mathrm{ph}}}+\tilde{n}_{i},
\label{derotate_chain}
\end{equation}
where $\mathbf{y}_{b}=[y_{b,1},\dots,y_{b,L}]^{\top}$ stacks the outputs of~\eqref{scalar_ch} over one block and $\tilde{n}_{i}\sim\mathcal{CN}(0,1)$ is the de-rotated noise.
Since each symbol follows $\mu_{\mathrm{GL}}$ by~\eqref{product_law}, the symbol posterior module calculates the posterior probability and the posterior mean as
\vspace{-0.4em}
\begin{equation}
\begin{aligned}
p_{i}(a)=\frac{\mu_{\mathrm{GL}}(a)\,e^{-|\tilde{y}_{i}-\sqrt{\gamma_{b}}\,a|^{2}}}
{\sum_{a'\in\mathcal{X}_{\mathrm{GL}}}\mu_{\mathrm{GL}}(a')\,e^{-|\tilde{y}_{i}-\sqrt{\gamma_{b}}\,a'|^{2}}},
\hat{v}_{i}=\sum_{a\in\mathcal{X}_{\mathrm{GL}}}a\,p_{i}(a).
\end{aligned}
\label{symbol_posterior}
\end{equation}
The estimates $\hat{v}_{1},\dots,\hat{v}_{L}$ are converted into the real-valued vector $\hat{\mathbf{v}}\in\mathbb{R}^{2L}$ by stacking their real and imaginary parts, and the inverse flow applies~\eqref{flow_inverse} as $\hat{\mathbf{u}}=F_{\theta}^{-1}(\hat{\mathbf{v}})$.

The semantic decoder $G_{\psi}$ adopts the reverse structure of the encoder in Fig.~\ref{network}, where the classification head outputs $\widehat{\Theta}$ after the self-attention blocks and the upsampling stages with the PixelShuffle head outputting $\widehat{S}\in\mathbb{R}^{H\times W\times 3}$.
Since $\gamma_{b}^{\star}$ in~\eqref{P_star} varies with $\epsilon$, $n$ and $\mathbf{T}^{\star}$, the feature-wise linear modulation~(FiLM) conditioning adapts the decoder to the operating SNR through the feature modulation
\vspace{-0.4em}
\begin{equation}
\begin{aligned}
&\tilde{\mathbf{h}}_{\ell}=\big(\mathbf{1}+\mathbf{m}_{\ell}(\gamma_{b})\big)\odot\mathbf{h}_{\ell}+\mathbf{q}_{\ell}(\gamma_{b}),\\
&\big[\mathbf{m}_{\ell}(\gamma_{b}),\mathbf{q}_{\ell}(\gamma_{b})\big]=\mathrm{MLP}_{\psi}(\gamma_{b}),
\end{aligned}
\label{film_mod}
\end{equation}
where $\mathbf{h}_{\ell}$ denotes the feature map of the $\ell$-th decoder stage, $\odot$ denotes the element-wise product, and $\mathbf{m}_{\ell}$ and $\mathbf{q}_{\ell}$ are generated by a two-layer perceptron from $\gamma_{b}$ in dB.
During training, $\gamma_{b}$ is drawn from an interval around $\gamma_{b}^{\star}$ to cover the operating points of different $(\epsilon,n,\mathbf{T}^{\star})$.

With $\Omega^{\star}$ fixed, the covertness constraint of $(\mathcal{P}1)$ is satisfied and the remaining objective is to minimize the weighted distortion
\vspace{-0.4em}
\begin{equation}
\begin{aligned}
\mathcal{L}(\varphi,\psi)
&=\mathbb{E}\big[\varrho_{\mathrm{o}}(S,\widehat{S})\big]+\beta\,\mathbb{E}\big[\bar{\varrho}_{\mathrm{s}}(S,\widehat{\Theta})\big]\\
&=\mathbb{E}\Big[\frac{\|S-\widehat{S}\|^{2}}{3HW}\Big]-\beta\,\mathbb{E}\big[\log\hat{p}_{\psi}(\Theta\mid\hat{\mathbf{u}},\gamma_{b})\big],
\end{aligned}
\label{loss_fn}
\end{equation}
where $\hat{p}_{\psi}(\theta\mid\hat{\mathbf{u}},\gamma_{b})$ denotes the probability of class $\theta$ produced by the classification head and $\|S-\widehat{S}\|^{2}$ denotes the sum of the squared pixel errors.
The covert symbol mapping and the invertible flow contain no trainable parameters, and the gradient of~\eqref{loss_fn} reaches $\varphi$ through the covert symbol mapping by the straight-through estimator, which forwards the constellation point in~\eqref{tx_block} and backpropagates through the soft assignment of $v_{i}$ to $\mathcal{X}_{\mathrm{GL}}$.
The complete training procedure is summarized in Algorithm~\ref{alg:train}.

\begin{algorithm}[!t]
\caption{Physical-Layer Design and End-to-End Training}
\label{alg:train}
\begin{algorithmic}[1]
\Require training set $\{(S,\Theta)\}$, covertness budget $\epsilon$, session length $n$, outage tolerance $\delta$, PAPR constraint $\Pi$, power budget $P_{\max}$, weight $\beta$, learning rate $\zeta$
\Ensure encoder $E_{\varphi}$, decoder $G_{\psi}$, configuration $\Omega^{\star}$
\State Select the largest $m$ with $\max_{k}\rho_{k}\le\Pi$ and form $(\mathcal{X}_{\mathrm{GL}},\mu_{\mathrm{GL}},t)$ by Proposition~\ref{prop_gl}
\State Obtain $\mathbf{T}^{\star}$ by~\eqref{place_smooth} with~\eqref{grad_lambda}--\eqref{grad_column}, and set $\mathbf{w}^{\star}=\mathbf{w}_{b}(\mathbf{T}^{\star})$, $P^{\star}$ and $\gamma_{b}^{\star}$ by~\eqref{P_star}
\State Initialize $\varphi$ and $\psi$ and generate $\theta$ from the synchronized seed
\While{not converged}
    \State compute $\mathbf{u}$ by~\eqref{layer_norm} and $\mathbf{v}$ by~\eqref{flow_def}, form $\mathbf{x}$ by~\eqref{occupancy}--\eqref{tx_block}
    \State draw $\gamma_{b}$ around $\gamma_{b}^{\star}$, obtain $y_{b,k}$ by~\eqref{scalar_ch} and recover $\hat{v}_{i}$ by~\eqref{derotate_chain}--\eqref{symbol_posterior}
    \State compute $\hat{\mathbf{u}}=F_{\theta}^{-1}(\hat{\mathbf{v}})$ by~\eqref{flow_inverse} and decode $(\widehat{S},\widehat{\Theta})=G_{\psi}(\hat{\mathbf{u}},\gamma_{b})$ with~\eqref{film_mod}
    \State update $(\varphi,\psi)\leftarrow(\varphi,\psi)-\zeta\nabla\mathcal{L}(\varphi,\psi)$ with~\eqref{loss_fn}
\EndWhile
\State \Return $E_{\varphi}$, $G_{\psi}$, $\Omega^{\star}$
\end{algorithmic}
\vspace{-0.2em}
\end{algorithm}


\section{Numerical and Simulation Results}\label{sec:results}

\begin{figure}[t]
    \centering
    \subfigure[Divergence versus the null scale $s$ at $\gamma=0$~dB.]{%
        \includegraphics[width=.5\linewidth]{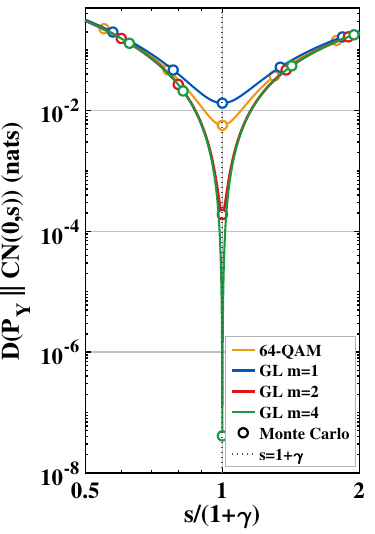}%
        \label{fig:val_lemmas_a}%
    }%
    \hfill%
    \subfigure[Non-Gaussianness $\Psi_{\mu}$ versus $\alpha$.]{%
        \includegraphics[width=.5\linewidth]{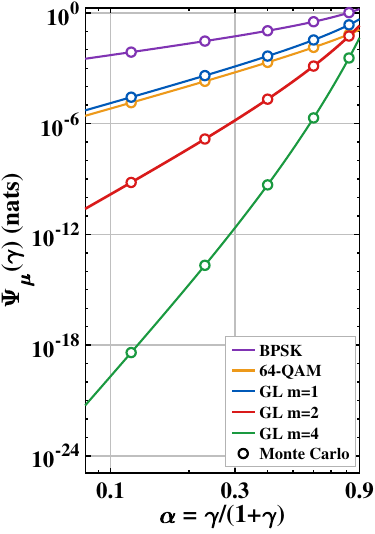}%
        \label{fig:val_lemmas_b}%
    }%
    \vspace{-0.5em}
    \caption{Validation of Lemma~\ref{lemma_ng} and Lemma~\ref{lemma_expansion} for BPSK, 64-QAM and the GL constellations with $m=1,2,4$ rings with Monte Carlo markers.}
    \label{fig:val_lemmas}
    \vspace{-1.5em}
\end{figure}


\subsection{Simulation Settings}\label{sec: sim_settings}

\subsubsection{Dataset}
\vspace{-0.3em}
A $100$-class ImageNet subset~\cite{deng2009imagenet} is adopted, comprising $1000$ training, $50$ validation, and $50$ test images per class.
All images are cropped to $H=W=128$, and the performance is measured by the peak signal-to-noise ratio (PSNR), the structural similarity index (SSIM) and the top-1 accuracy of a frozen ResNet-50 classifier~\cite{he2016deep}.

\subsubsection{Training Parameters}
\vspace{-0.3em}
All networks are trained by the AdamW optimizer with a learning rate of $3\times10^{-4}$ and a batch size of $64$ for $100$ epochs, and the loss weight in~\eqref{loss_fn} is $\beta=0.02$.

\subsubsection{Benchmarks}
\vspace{-0.3em}
DeepJSCC-Q~\cite{tung2022deepjsccq} with $64$-QAM, the masked VQ-VAE scheme~\cite{hu2023robust} with a learned codebook, SwinJSCC~\cite{yang2025swinjscc} with analog symbols and the MA-covert scheme~\cite{hu2026movable} with Gaussian signaling over the array designed for the legitimate link are adopted for comparison.
All schemes use the same $L$, the MA array, $\Pi$, and the scrambling of Section~\ref{sec: structure_layer}, and each transmits at the maximum power admitted by Theorem~\ref{theorem_rate} for its respective symbol distribution under the same $(\epsilon,n,\delta)$.

\subsubsection{System Parameters and Network Settings}
\vspace{-0.3em}
Unless otherwise stated, Alice is equipped with $M=16$ MAs that start from a $4\times4$ half-wavelength grid and move within a square region $\mathcal{A}$ of side $8\lambda$ with $d_{\min}=\lambda/2$, while Bob and Willie are equipped with $N=4$ and $M_{w}=4$ antennas.
Each link consists of $L_{p}=10$ paths with angles uniformly distributed over $[0,2\pi)$, and the transmit power in~\eqref{P_star} is measured relative to the normalized noise powers $\sigma_{b}^{2}=\sigma_{w}^{2}=1$ with an inactive budget $P_{\max}$.
The covertness parameters are $\epsilon=0.1$, $n=10^{6}$ and $\delta=10^{-2}$ with the PAPR constraint $\Pi=6$~dB.
The detailed network layer configuration is listed in Table~\ref{tab:network_config}.

\begin{table}[t]
\centering
\footnotesize
\caption{Layer Configuration of the Proposed Network}
\label{tab:network_config}
\vspace{-0.7em}
\renewcommand{\arraystretch}{1.0}
\setlength{\tabcolsep}{3pt}
\begin{tabular}{>{\centering\arraybackslash}m{1.45cm}ccc}
\hline
\textbf{Module} & \textbf{Layer} & \textbf{Width} & \textbf{Act.} \\
\hline
Semantic Encoder &
\makecell{Conv Stem ($5\times5$, stride $2$) \\ Downsample Stage $\times2$ \\ Self-Attention ($8$ heads) \\ Proj + LayerNorm} &
\makecell{$224$ \\ $448/672$ \\ $672$ \\ $128$} &
\makecell{GELU \\ GELU \\ GELU \\ --} \\
\hline
Invertible Flow &
\makecell{Sign Flip, Hadamard Mix, \\ Permutation} &
$2L$ & -- \\
\hline
Covert Symbol Mapping &
\makecell{Rank to Ring ($m=2$) \\ Phase Quantise ($M_{\mathrm{ph}}=32$) \\ Scramble + Interleave} &
\makecell{$L$ \\ $L$ \\ $L$} &
\makecell{-- \\ -- \\ --} \\
\hline
Semantic Decoder &
\makecell{Conv In, Self-Attention $\times2$ \\ Classify Head \\ Upsample Stage $\times2$ \\ PixelShuffle Head \\ FiLM Conditioning} &
\makecell{$672$ \\ $100$ \\ $448/224$ \\ $3$ \\ $128$} &
\makecell{GELU \\ Softmax \\ GELU \\ Sigmoid \\ GELU} \\
\hline
\end{tabular}
\vspace{-0.8em}
\end{table}
\subsection{Verification of Theoretical Results}
\vspace{-0.3em}

First, we validate Lemma~\ref{lemma_ng} and Lemma~\ref{lemma_expansion} through Monte Carlo simulations in Fig.~\ref{fig:val_lemmas}.
As shown in Fig.~\ref{fig:val_lemmas_a}, the Monte Carlo estimates of the divergence from the composite null hypothesis at $\gamma=0$~dB fall on the theoretical curves and the minimum value occurs at $s=1+\gamma$, with the value being  $\log(1+\gamma)-I_{\mu}(\gamma)$, which validates~\eqref{non_gaussianness} in Lemma~\ref{lemma_ng}.
The curves overlap for $s\neq1+\gamma$ since the scale mismatch contributes $D(\mathcal{CN}(0,1+\gamma)\|\mathcal{CN}(0,s))$ regardless of $\mu$, whereas the valley floor set by the constellation decreases with $m$.
In Fig.~\ref{fig:val_lemmas_b}, the Monte Carlo estimates of $\Psi_{\mu}(\gamma)$ match the theoretical curves that follow straight lines of slope $t+1$ against $\alpha$ on logarithmic axes, which validates~\eqref{low_snr} in Lemma~\ref{lemma_expansion} with $t=1$ for BPSK, $t=3$ for 64-QAM and $t=4m-1$ for the GL constellations by Proposition~\ref{prop_gl}.
At $\Psi_{\mu}=10^{-6}$, the admissible $\alpha$ increases from $0.06$ for 64-QAM to $0.72$ for GL $m=4$, and the convergence of the curves as $\alpha\to1$ confines the gain of the matching order to the low-SNR regime enforced by the budget $2\epsilon^{2}/n$ in~\eqref{max_power}.

\begin{figure}[t]
    \centering
    \includegraphics[width=0.84\linewidth]{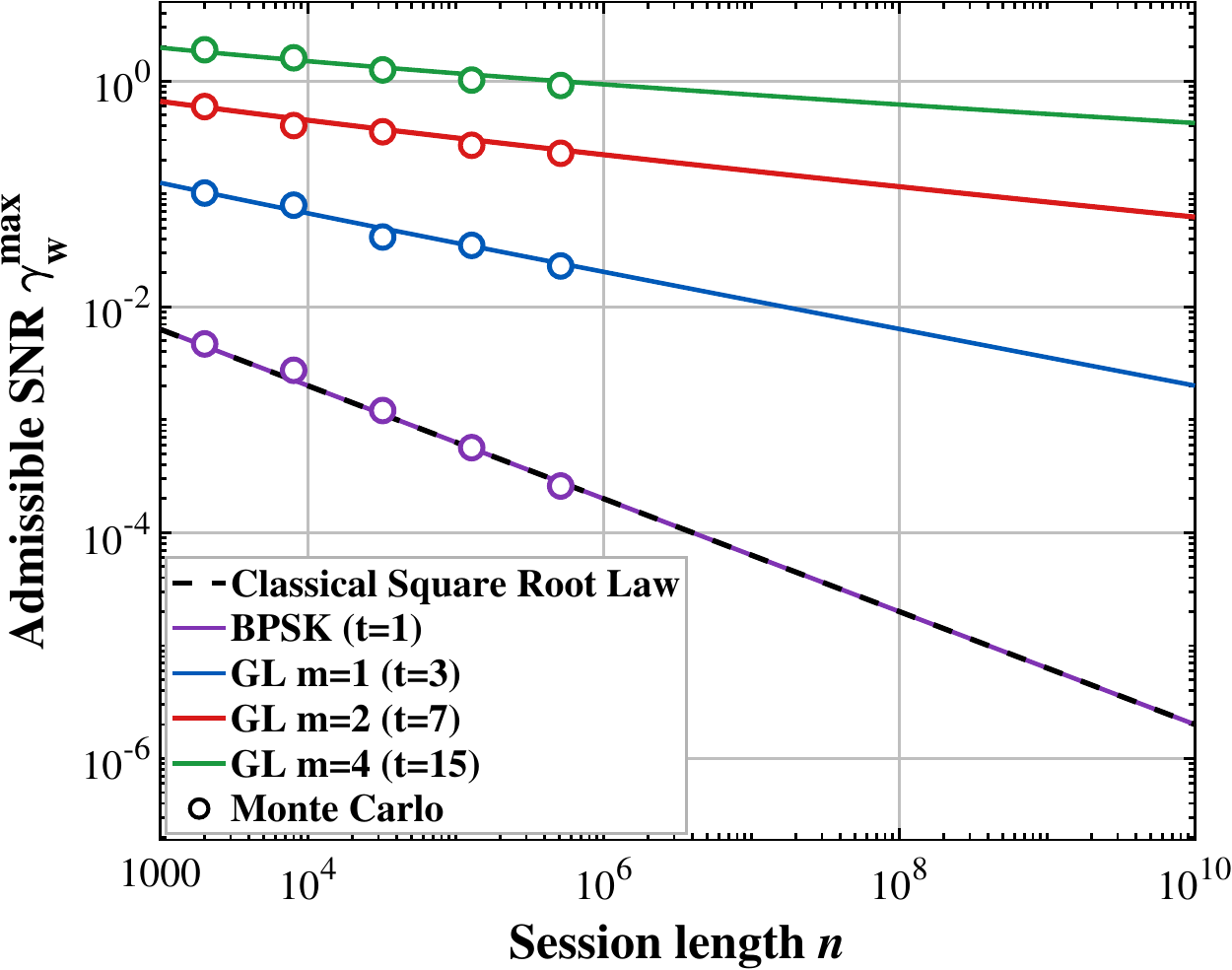}
    \vspace{-0.5em}
    \caption{Admissible warden SNR versus the session length $n$ under $\epsilon=0.1$ for the validation of Theorem~\ref{theorem_rate}.}
    \label{fig:val_theorem1}
    \vspace{-0.8em}
\end{figure}

As shown in Fig.~\ref{fig:val_theorem1}, Theorem~\ref{theorem_rate} is validated through the admissible warden SNR $\gamma_{w}^{\max}$ in~\eqref{max_power} against the session length $n$ under $\epsilon=0.1$.
The Monte Carlo markers denote the warden SNR at which the likelihood-ratio test at Willie over a session of $n$ symbols attains the minimum risk corresponding to the leakage budget $2\epsilon^{2}$ in~\eqref{covert_constraint}, and match the theoretical curves for all simulated constellations.
The BPSK curve with $t=1$ coincides with the classical square-root law $2\epsilon/\sqrt{n}$~\cite{bash2013limits}, which confirms that~\eqref{max_snr_rate} reduces to the classical law at $t=1$.
Since $\gamma_{w}^{\max}$ decays as $n^{-1/(t+1)}$ by~\eqref{low_snr}, a hundredfold increase in $n$ lowers the admissible SNR by $10$~dB under the classical law but by $2.5$~dB for GL $m=2$, and at $n=10^{6}$ the GL constellations with $m=2$ and $m=4$ admit warden SNRs of $-6.5$~dB and $-0.3$~dB compared with $-37$~dB under the classical law.

\subsection{Performance Analysis and Comparison}

\begin{figure} 
  \centering
  \includegraphics[width=0.8\linewidth]{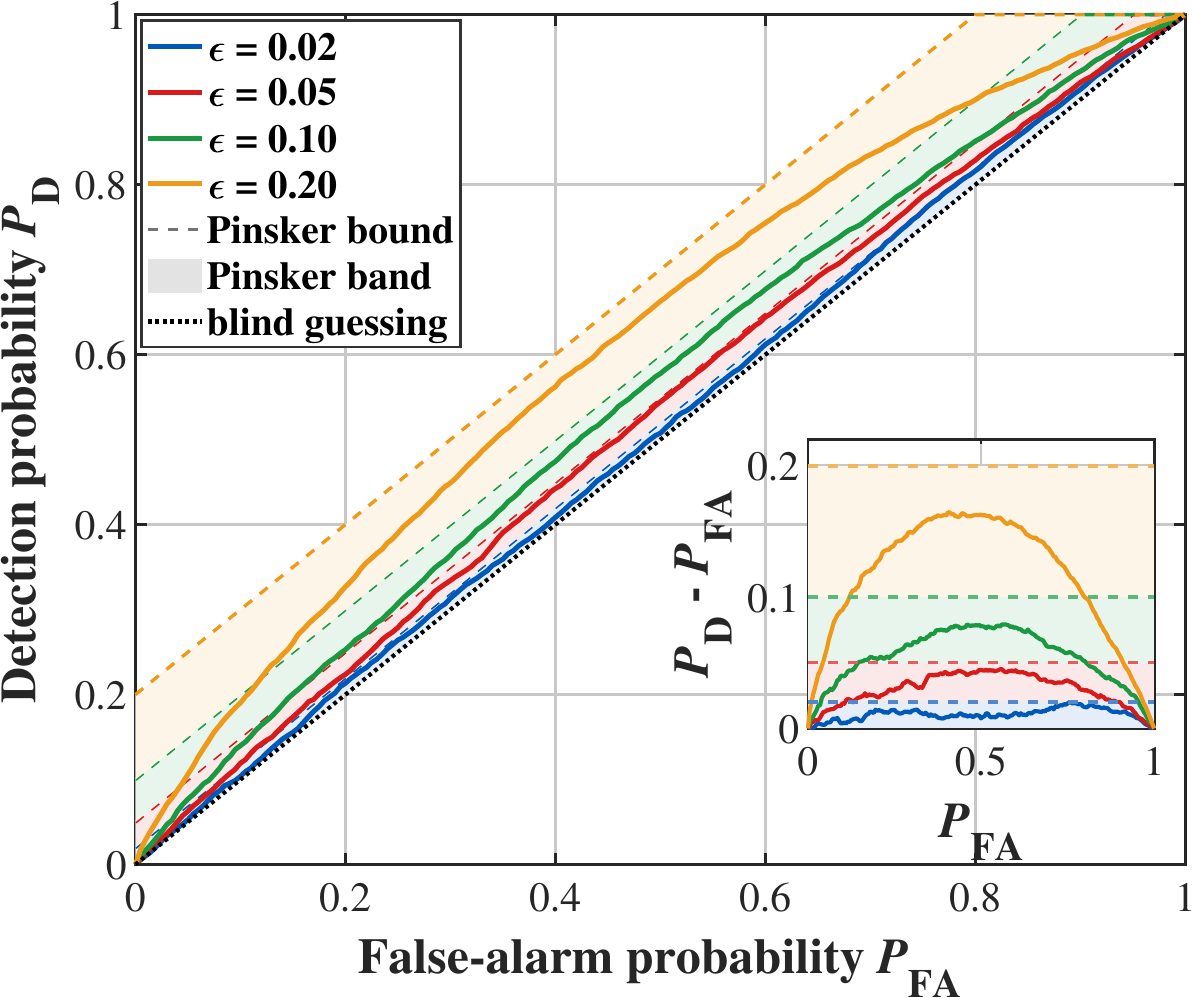}
    \vspace{-0.5em}
    \caption{Detection probability versus false-alarm probability of the optimal test at Willie under the admissible warden SNR for various $\epsilon$.}
    \label{fig:willie_roc}
    \vspace{-1.5em}
\end{figure}

Fig.~\ref{fig:willie_roc} depicts the detection probability $P_{\mathrm{D}}$ against the false-alarm probability $P_{\mathrm{FA}}$ of the likelihood-ratio test at Willie for sessions of $n=5000$ symbols from $\mathcal{X}_{\mathrm{GL}}$ with $m=2$, where the warden SNR equals the admissible value $\Psi_{\mu}^{-1}(2\epsilon^{2}/n)$ in~\eqref{max_power} for each $\epsilon$.
It can be observed that the receiver operating characteristic~(ROC) remains within the Pinsker band $P_{\mathrm{D}}\le P_{\mathrm{FA}}+\epsilon$ implied by~\eqref{minimax_bound} for all four budgets.
Since the likelihood-ratio test is optimal under the composite null hypothesis in~\eqref{hypotheses}, an arbitrary detector at Willie gains no more than $\epsilon$ under the admissible power of Theorem~\ref{theorem_rate} and the Gaussian-matched constellation, which verifies the effectiveness of the covert design.

\begin{figure}[t]
    \centering
    \subfigure[PSNR versus $M$.]{%
        \includegraphics[width=.48\linewidth]{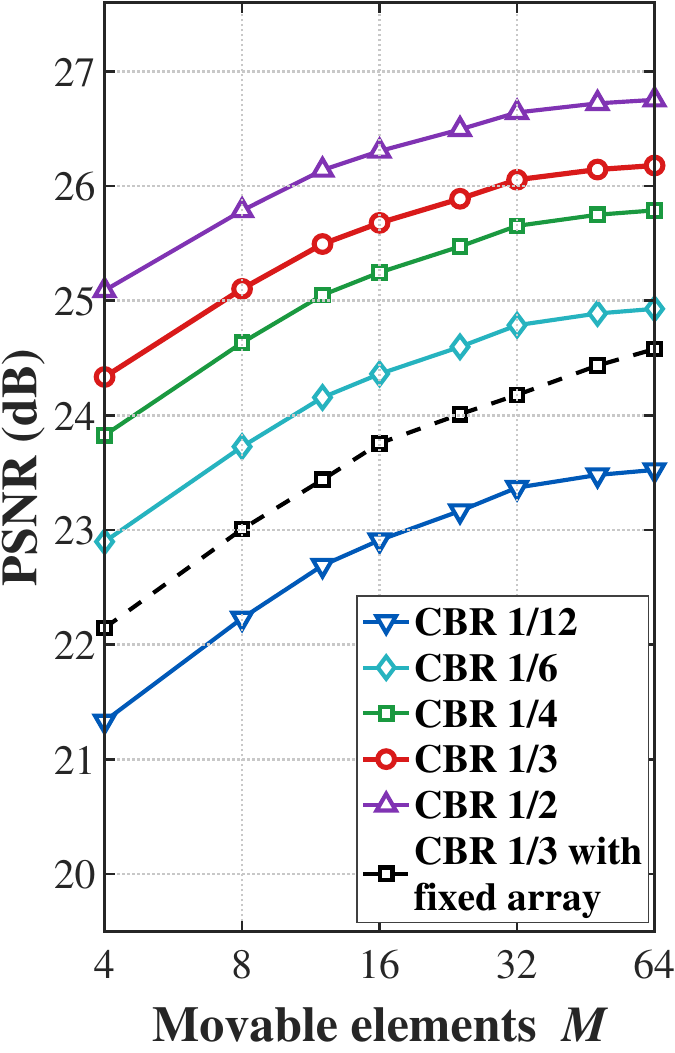}%
        \label{array_size}%
    }%
    \hfill%
    \subfigure[PSNR versus $A$.]{%
        \includegraphics[width=.48\linewidth]{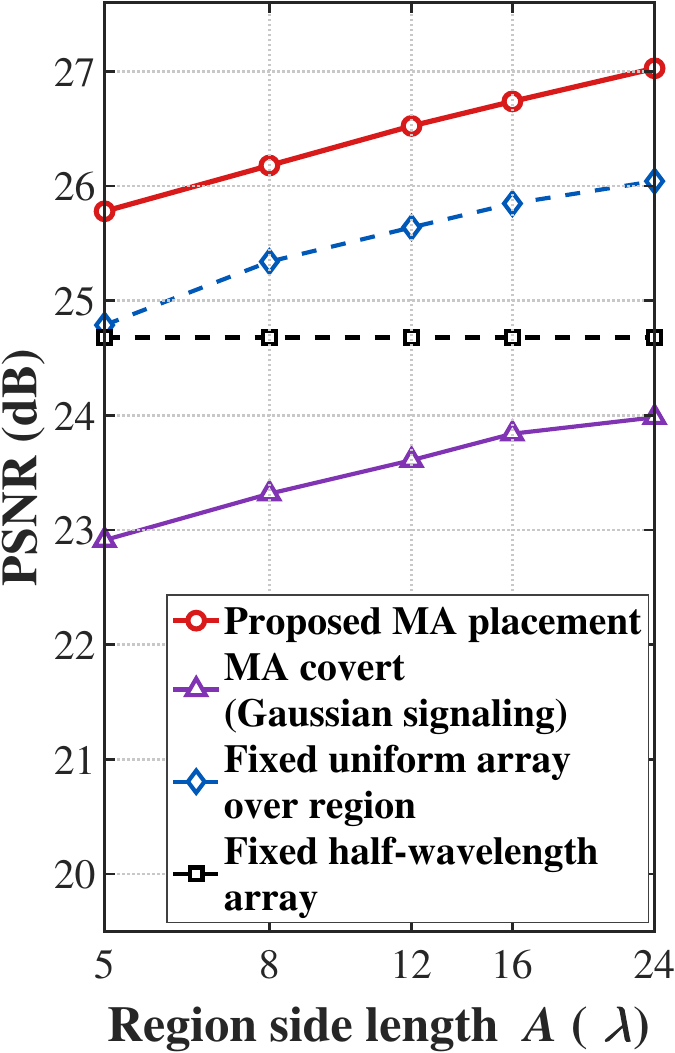}%
        \label{array_region}%
    }%
    \vspace{-0.5em}
    \caption{PSNR gains of the MA array over the fixed arrays under the same covertness budget.}
    \label{fig:ma_gain}
    \vspace{-1.3em}
\end{figure}

Fig.~\ref{array_size} illustrates the PSNR versus the number of MAs $M$ under different channel bandwidth ratios (CBRs), where the scheme of fixed array places the $M$ antennas on a half-wavelength grid under the same covertness budget.
The PSNR increases with $M$ under each CBR with diminishing returns, since the array term $\Delta/u_{1-\delta}$ in~\eqref{max_snr_rate} increases $\gamma_{b}^{\star}$ while the information per symbol is bounded by $\log K$ at a fixed CBR.
Compared with the fixed array, the MA placement of Corollary~\ref{cor_place} yields a gain of $1.6$ to $2.2$~dB at CBR $1/3$ over the considered range of $M$.

Fig.~\ref{array_region} further illustrates the PSNR versus the side length $A$ of the region $\mathcal{A}$ with $M=64$ at CBR $1/3$, where two fixed arrays serve as references. 
The uniform array distributes the $M$ antennas evenly over the region of side $A$, whereas the half-wavelength array places them on a $\lambda/2$ grid independent of $A$. 
The PSNR of the MA array increases with $A$ with diminishing returns, since a larger region admits a more flexible placement position adjustment that reduces the correlation in $\bar{\mathbf{R}}_{w}(\mathbf{T})$. 
Compared with the uniform array over the same region, the MA placement of Corollary~\ref{cor_place} attains a gain of $0.8$ to $1.1$~dB, and the gain over the half-wavelength array reaches $2.2$~dB, which indicates that the improvement stems from the optimized positions rather than from the enlarged region. 
The MA-covert benchmark with Gaussian signaling benefits from the region in the same manner but remains over $2.8$~dB below the proposed scheme due to its lower matching order.

\begin{figure*}[!t]
    \centering
    \subfigure[Comparison of PSNR.]{%
        \includegraphics[width=.33\textwidth]{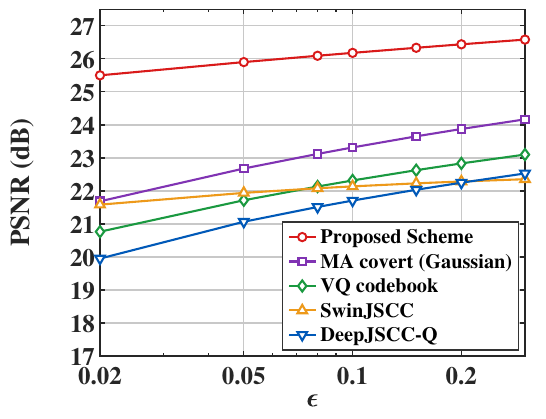}%
        \label{fig:compare_psnr}%
    }%
    \hfill%
    \subfigure[Comparison of SSIM.]{%
        \includegraphics[width=.33\textwidth]{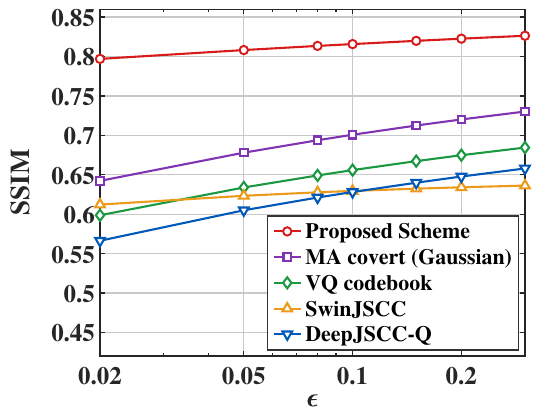}%
        \label{fig:compare_ssim}%
    }%
    \hfill%
    \subfigure[Comparison of semantic accuracy.]{%
        \includegraphics[width=.33\textwidth]{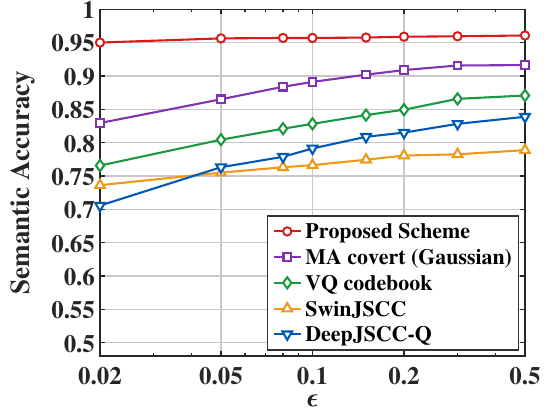}%
        \label{fig:compare_acc}%
    }%
    \vspace{-0.5em}
    \caption{Performance comparison of the proposed scheme and the benchmark schemes DeepJSCC-Q~\cite{tung2022deepjsccq}, the masked VQ-VAE scheme~\cite{hu2023robust}, SwinJSCC~\cite{yang2025swinjscc} and the MA-covert scheme~\cite{hu2026movable} under various covertness budgets $\epsilon$.}
    \label{fig:compare_all}
    \vspace{-0.8em}
\end{figure*}

Fig.~\ref{fig:compare_all} compares the proposed scheme with the benchmarks versus the covertness budget $\epsilon$, where each scheme operates at its own admissible power under Theorem~\ref{theorem_rate}.
As shown in Fig.~\ref{fig:compare_psnr} and Fig.~\ref{fig:compare_ssim}, the proposed scheme improves the PSNR by over $3.8$~dB and the SSIM by $24\%$ over the strongest benchmark at $\epsilon=0.02$.
The advantage under strict budgets arises from the matching order $t=7$, under which $\gamma_{b}^{\star}$ in~\eqref{max_snr_rate} scales with $\epsilon^{1/4}$ instead of $\epsilon^{1/2}$ for the benchmarks with $t=3$.
Moreover, Fig.~\ref{fig:compare_acc} shows that the semantic accuracy of the proposed scheme remains above $95\%$ over the entire range of $\epsilon$ with a significant performance gap over the benchmarks, which demonstrates that the task-relevant semantics are well preserved under strict covertness constraints.

\begin{figure*}[!t]
    \centering
    \setlength{\subfigtopskip}{0pt}
    \setlength{\subfigcapskip}{0pt}
    \setlength{\subfigbottomskip}{0pt}
    \def\PairSpan{\dimexpr\textwidth-2.2bp\relax}
    \makebox[\textwidth][c]{%
        \subfigure[PSNR versus the legitimate SNR.]{%
            \resizebox{0.437308868502\PairSpan}{!}{%
                \includegraphics{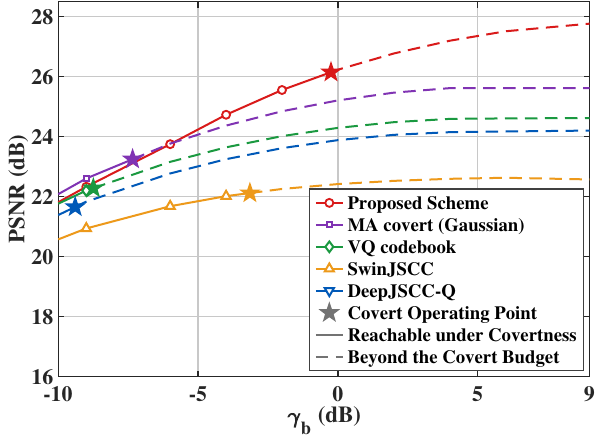}%
            }%
            \label{fig:operating_curves}%
        }%
        \hspace{2bp}%
        \subfigure[Reconstructed images at the respective covert operating points.]{%
            \raisebox{0.040061157413\PairSpan}{%
                \resizebox{0.562691131498\PairSpan}{!}{%
                    \includegraphics{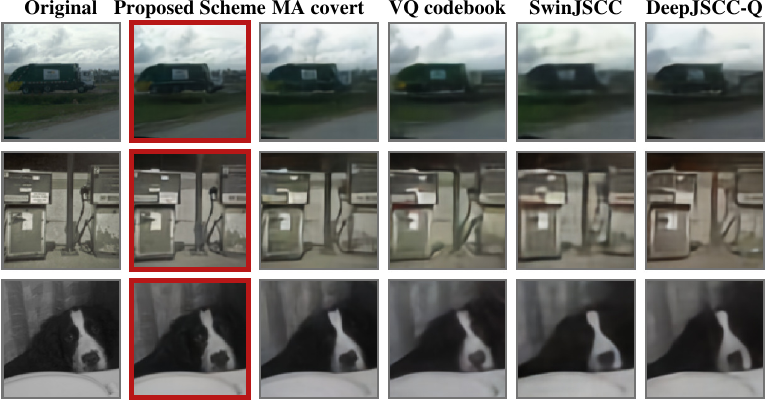}%
                }%
            }%
            \label{fig:visual_grid}%
        }%
    }%
    \vspace{-0.5em}
    \caption{Reconstruction performance of the proposed scheme and the benchmark schemes under the covertness constraint.}
    \label{fig:operating_and_visual}
    \vspace{-1.5em}
\end{figure*}

Fig.~\ref{fig:operating_curves} plots the PSNR versus the legitimate SNR $\gamma_{b}$, where each solid segment ends at the covert operating point $\gamma_{b}^{\star}$ and the dashed segment is unreachable under the covertness constraint.
The proposed scheme attains the highest PSNR at each $\gamma_{b}$ above $-6$~dB and the gap increases with the SNR, since the benchmarks saturate at lower PSNR ceilings while the proposed curve continues to increase.
In addition, the operating point of the proposed scheme lies $2.9$ to $9.1$~dB higher than those of the benchmarks by virtue of $t=7$ in~\eqref{max_snr_rate}, and the higher curve together with the higher operating point accounts for the PSNR gap of over $2.8$~dB at the respective operating points.
In the reconstructed images of Fig.~\ref{fig:visual_grid}, the proposed scheme restores sharper contours and textures in the vehicle and the dispenser, whereas the benchmarks blur the edges and smear the details.



\section{Conclusion}\label{sec:conclusion}
\vspace{-0.3em}

In this paper, we proposed a covert digital SemCom system assisted by an MA array in the presence of a warden with unknown noise level.
The maximum covert rate was characterized against the optimal detector with unknown noise level in closed-form, which scales the admissible power with the moment-matching order of the constellation.
Based on the derived rate, a GL ring constellation and a matching-order-aware MA placement were developed, and a covert symbol mapping together with an invertible flow was designed to attain the covert rate while conveying the semantics under the derived distortion bound.
Simulation results demonstrated that the proposed system achieves a superior reconstruction quality under covertness constraints, which provides a practical route to conceal semantic transmission on digital transceivers.

\begin{appendices}

\section{Proofs of Lemma~\ref{lemma_ng} and Lemma~\ref{lemma_expansion}}\label{app:lemmas}
\vspace{-0.3em}

Let $p_{Y}$ denote the probability density function~(PDF) of $\mathbb{P}_{Y}$.
For $s>0$, the divergence in~\eqref{non_gaussianness} is expressed as
\vspace{-0.4em}
\begin{equation}
\begin{aligned}
D\big(\mathbb{P}_{Y}\,\big\|\,\mathcal{CN}(0,s)\big)
&=\int p_{Y}(y)\log\frac{p_{Y}(y)}{(\pi s)^{-1}e^{-|y|^{2}/s}}\,\mathrm{d}y\\
&=-h(Y)+\log(\pi s)+\frac{1+\gamma}{s},
\end{aligned}
\label{app_kl_expand}
\end{equation}
where $\mathbb{E}[|Y|^{2}]=\gamma\,\mathbb{E}_{\mu}[|X|^{2}]+\mathbb{E}[|W|^{2}]=1+\gamma$ follows from the independence between $X$ and $W$.
The derivative of~\eqref{app_kl_expand} with respect to $s$ equals $(s-1-\gamma)/s^{2}$, which identifies $s^{\star}=1+\gamma$ as the unique minimizer.
Substituting $s^{\star}$ together with the chain rule $h(Y)=h(Y\,|\,X)+I(X;Y)$ and $h(Y\,|\,X)=h(W)=\log(\pi e)$ into~\eqref{app_kl_expand} yields
\vspace{-0.4em}
\begin{equation}
\Psi_{\mu}(\gamma)=-h(Y)+\log\big(\pi(1+\gamma)\big)+1=\log(1+\gamma)-I_{\mu}(\gamma).
\label{app_psi_closed}
\end{equation}
Since the divergence is invariant to a common invertible transformation of both distributions, $\Psi_{\mu}(\gamma)=D(\mathbb{P}_{Z}\|\mathcal{CN}(0,1))$ holds for the normalized observation $Z=Y/\sqrt{1+\gamma}$.
For the monotonicity, the derivative of~\eqref{app_psi_closed} equals $1/(1+\gamma)-\mathrm{mmse}_{\mu}(\gamma)$ by the relation between the mutual information and the minimum mean-square error (MMSE) $\mathrm{mmse}_{\mu}(\gamma)$ of estimating $X$ from $Y$~\cite{guo2005mutual}, which is positive since the MMSE of a unit-power input is upper bounded by $1/(1+\gamma)$ with equality for the Gaussian input.
Together with $I_{\mu}(0)=0$, $\Psi_{\mu}(0)=0$ holds, which completes the proof of Lemma~\ref{lemma_ng}.

For the proof of Lemma~\ref{lemma_expansion}, let $\phi_{v}(z)=(\pi v)^{-1}e^{-|z|^{2}/v}$ denote the density of $\mathcal{CN}(0,v)$ with $\phi=\phi_{1}$, under which the complex Hermite polynomials are orthogonal with $\mathbb{E}_{\phi}[H_{p,q}H^{*}_{p',q'}]=p!\,q!\,\mathbf{1}\{(p,q)=(p',q')\}$.
By the definition of $\alpha$, the normalized observation $Z=\sqrt{\alpha}\,X+\sqrt{\bar{\alpha}}\,W$ with $\bar{\alpha}=1-\alpha$ has the density $p_{Z}(z)=\mathbb{E}_{X}[\phi_{\bar{\alpha}}(z-\sqrt{\alpha}X)]$, and the chi-square divergence $\chi^{2}_{\mu}(\gamma)=\int p_{Z}^{2}(z)/\phi(z)\,\mathrm{d}z-1$ involves the Gaussian overlap integral with its Mehler expansion~\cite{ismail2016analytic}
\vspace{-0.4em}
\begin{equation}
\begin{aligned}
\kappa_{\alpha}(a,a')
&=\int\frac{\phi_{\bar{\alpha}}(z-\sqrt{\alpha}a)\,\phi_{\bar{\alpha}}(z-\sqrt{\alpha}a')}{\phi(z)}\,\mathrm{d}z\\
&=\sum_{p,q\ge 0}\frac{\alpha^{p+q}}{p!\,q!}H_{p,q}(a)H^{*}_{p,q}(a').
\end{aligned}
\label{app_mehler}
\end{equation}
With an independent copy $\tilde{X}$ of $X$, the chi-square divergence is expanded as
\vspace{-1em}
\begin{equation}
\begin{aligned}
1+\chi^{2}_{\mu}(\gamma)
&\overset{(a)}{=}\mathbb{E}_{X,\tilde{X}}\big[\kappa_{\alpha}(X,\tilde{X})\big]\\
&=\sum_{p,q\ge 0}\frac{\alpha^{p+q}}{p!\,q!}\big|\mathbb{E}_{\mu}[H_{p,q}(X)]\big|^{2}\\
&\overset{(b)}{=}1+2\kappa_{t}\,\alpha^{t+1}\big(1+O(\alpha)\big),
\end{aligned}
\label{app_chi2}
\end{equation}
where $(a)$ follows by substituting $p_{Z}$ into $\chi^{2}_{\mu}(\gamma)$ and applying~\eqref{app_mehler} together with the independence between $X$ and $\tilde{X}$, and $(b)$ follows from $H_{0,0}=1$ together with the matching order in~\eqref{matching_order} and the coefficient in~\eqref{kappa_def}.
Let $r(z)=p_{Z}(z)/\phi(z)-1$ denote the centered density ratio with $\mathbb{E}_{\phi}[r]=0$ and $\mathbb{E}_{\phi}[r^{2}]=\chi^{2}_{\mu}(\gamma)$.
The output non-Gaussianness is then expanded as
\vspace{-0.4em}
\begin{equation}
\begin{aligned}
&\Psi_{\mu}(\gamma)
\overset{(c)}{=}\mathbb{E}_{\phi}\big[(1+r)\log(1+r)\big]\\
&=\mathbb{E}_{\phi}\Big[r+\frac{r^{2}}{2}\Big]+O\big(\mathbb{E}_{\phi}[|r|^{3}]\big)\\
&\overset{(d)}{=}\frac{1}{2}\chi^{2}_{\mu}(\gamma)+O\big(\alpha^{3(t+1)/2}\big)
\overset{(e)}{=}\kappa_{t}\,\alpha^{t+1}\big(1+O(\alpha)\big),
\end{aligned}
\label{app_psi_expand}
\end{equation}
where $(c)$ expresses $D(\mathbb{P}_{Z}\|\phi)$ as an expectation under $\phi$ and applies the Taylor expansion $(1+r)\log(1+r)=r+r^{2}/2+O(|r|^{3})$.
Step $(d)$ uses $\mathbb{E}_{\phi}[|r|^{3}]=O(\alpha^{3(t+1)/2})$, which follows from the Hermite expansion of $r$ that starts at degree $t+1$ with coefficients proportional to $\alpha^{(p+q)/2}$ and from the hypercontractivity of Gaussian chaos~\cite{janson1997gaussian}.
Step $(e)$ follows from~\eqref{app_chi2} and from $3(t+1)/2\ge t+2$ for $t\ge 1$, which holds for every zero-mean $\mu$.
This completes the proof of Lemma~\ref{lemma_expansion}.


\section{Proof of Theorem~\ref{theorem_rate}}\label{app:theorem_rate}
\vspace{-0.3em}

Under fixed $\mathbf{T}$ and $\mathbf{w}$, $\gamma_{w}=\bar{P}U$ is constant over the session, and the i.i.d. symbols together with the i.i.d. noise yield $\mathbb{P}^{(n)}_{1}=\mathbb{P}_{Y}^{\otimes n}$ with $\gamma=\bar{P}U$ while $\mathcal{CN}(\mathbf{0},s\mathbf{I}_{n})$ is the $n$-fold product of $\mathcal{CN}(0,s)$.
The additivity of the divergence over products and~\eqref{non_gaussianness} simplify the session leakage as
\vspace{-0.4em}
\begin{equation}
\begin{aligned}
\mathcal{L}_{n}
&=\inf_{s>0}D\big(\mathbb{P}_{Y}^{\otimes n}\,\big\|\,\mathcal{CN}(0,s)^{\otimes n}\big)\\
&=n\inf_{s>0}D\big(\mathbb{P}_{Y}\,\big\|\,\mathcal{CN}(0,s)\big)=n\,\Psi_{\mu}\big(\bar{P}U\big).
\end{aligned}
\label{app_leak_chain}
\end{equation}
Since $\Psi_{\mu}^{-1}$ exists and is strictly increasing on $[0,\infty)$ by Lemma~\ref{lemma_ng}, the covertness probability equals
\vspace{-0.4em}
\begin{equation}
\begin{aligned}
\Pr\big\{\mathcal{L}_{n}\le 2\epsilon^{2}\big\}
&=\Pr\Big\{U\le\frac{1}{\bar{P}}\Psi_{\mu}^{-1}\Big(\frac{2\epsilon^{2}}{n}\Big)\Big\}\\
&=F_{U}\Big(\frac{1}{\bar{P}}\Psi_{\mu}^{-1}\Big(\frac{2\epsilon^{2}}{n}\Big)\Big).
\end{aligned}
\label{app_prob_chain}
\end{equation}
With the cumulative distribution function~(CDF) $F_{U}$ of $U$, the quantile $u_{1-\delta}=\inf\{x:F_{U}(x)\ge 1-\delta\}$ and the right continuity of $F_{U}$, the probability in~\eqref{app_prob_chain} is no less than $1-\delta$ under the necessary and sufficient condition $\Psi_{\mu}^{-1}(2\epsilon^{2}/n)/\bar{P}\ge u_{1-\delta}$, which is rearranged into~\eqref{max_power}.
The legitimate SNR $\gamma_{b}^{\star}=\Delta\bar{P}^{\star}$ in~\eqref{max_snr_rate} follows from~\eqref{avg_snr} since $\Delta$ is independent of $P$.

For the covert rate, substituting $\alpha=\gamma/(1+\gamma)=\gamma(1+O(\gamma))$ into~\eqref{low_snr} and inverting the expansion yield
\vspace{-0.5em}
\begin{align}
\Psi_{\mu}(\gamma)&=\kappa_{t}\,\gamma^{t+1}\big(1+O(\gamma)\big), \label{psi_forward} \\
\Psi_{\mu}^{-1}(x)&=\Big(\frac{x}{\kappa_{t}}\Big)^{\frac{1}{t+1}}\Big(1+O\big(x^{\frac{1}{t+1}}\big)\Big),\quad x\to 0. \label{app_inverse}
\end{align}
With~\eqref{non_gaussianness} and $\log(1+\gamma)=\gamma+O(\gamma^{2})$, the covert rate in~\eqref{max_snr_rate} is expanded with $x=2\epsilon^{2}/n$ as
\vspace{-0.4em}
\begin{equation}
\begin{aligned}
R_{\mathrm{c}}
&=\log\big(1+\gamma_{b}^{\star}\big)-\Psi_{\mu}\big(\gamma_{b}^{\star}\big)
=\gamma_{b}^{\star}+O\big((\gamma_{b}^{\star})^{2}\big)\\
&=\frac{\Delta}{u_{1-\delta}}\Big(\frac{2\epsilon^{2}}{\kappa_{t}\,n}\Big)^{\frac{1}{t+1}}\Big(1+O\big(n^{-\frac{1}{t+1}}\big)\Big),
\end{aligned}
\label{app_rate_chain}
\end{equation}
where $\gamma_{b}^{\star}=O(n^{-1/(t+1)})$ by~\eqref{app_inverse} and the ratio between $R_{\mathrm{c}}$ and the leading term tends to one as $n\to\infty$.
This completes the proof of Theorem~\ref{theorem_rate}.

\vspace{-0.5em}

\section{Proof of Theorem~\ref{theorem_distortion}}\label{app:theorem_distortion}
\vspace{-0.3em}

Let $\varXi$ denote the state of the scrambling and interleaving sequence, $\mathbb{P}_{1,k}$ denote the marginal distribution of the $k$-th observation under $H_{1}$ with the symbol distribution $\mu_{k}$, and let $\bar{\mu}=\frac{1}{n}\sum_{k=1}^{n}\mu_{k}$.
Since $\mathcal{CN}(\mathbf{0},s\mathbf{I}_{n})$ is a product distribution, the divergence in~\eqref{leakage} decomposes for $s>0$ as
\vspace{-0.4em}
\begin{equation}
\begin{aligned}
D\big(\mathbb{P}^{(n)}_{1}\,\big\|\,\mathcal{CN}(\mathbf{0},s\mathbf{I}_{n})\big)=\sum_{k=1}^{n}D\big(\mathbb{P}_{1,k}\,\big\|\,\mathcal{CN}(0,s)\big)+\mathrm{TC}\big(\mathbb{P}^{(n)}_{1}\big),
\end{aligned}
\label{app_tc}
\end{equation}
where the total correlation $\mathrm{TC}(\mathbb{P}^{(n)}_{1})\ge 0$ is the divergence between $\mathbb{P}^{(n)}_{1}$ and the product of its marginals.
Taking the infimum over $s$ yields
\vspace{-0.4em}
\begin{equation}
\begin{aligned}
\mathcal{L}_{n}
&\overset{(a)}{\ge}\inf_{s>0}\sum_{k=1}^{n}D\big(\mathbb{P}_{1,k}\,\big\|\,\mathcal{CN}(0,s)\big)\\
&\overset{(b)}{\ge}\sum_{k=1}^{n}\Psi_{\mu_{k}}\big(\bar{P}U\big)
\overset{(c)}{\ge}n\,\Psi_{\bar{\mu}}\big(\bar{P}U\big),
\end{aligned}
\label{app_leak_lower}
\end{equation}
where $(a)$ follows by dropping the total correlation, $(b)$ follows from~\eqref{non_gaussianness}, and $(c)$ follows from the Jensen inequality since $\mu\mapsto\Psi_{\mu}(\gamma)$ is convex due to the concavity of the mutual information in $\mu$.
Applying the quantile argument in~\eqref{app_prob_chain} to~\eqref{app_leak_lower} and multiplying by $\Delta$ bound the legitimate SNR as $\gamma_{b}=\Delta\bar{P}\le\frac{\Delta}{u_{1-\delta}}\Psi_{\bar{\mu}}^{-1}(2\epsilon^{2}/n)=\gamma_{b}^{\star}(\bar{\mu})$.

Let $I_{j}=I(S_{j};\widehat{S}_{j},\widehat{\Theta}_{j})$ denote the information delivered for the $j$-th object.
Since the reconstructions are functions of $(\mathbf{y}_{b},\varXi)$ and the objects are independent of each other and of $\varXi$, the data processing inequality yields
\vspace{-0.4em}
\begin{equation}
\begin{aligned}
\sum_{j=1}^{J}I_{j}
&\le I\big(S_{1},\dots,S_{J};\mathbf{y}_{b}\,\big|\,\varXi\big)
\le I\big(\mathbf{x};\mathbf{y}_{b}\,\big|\,\varXi\big)\\
&\overset{(d)}{\le}I\big(\mathbf{x};\mathbf{y}_{b}\big)
\overset{(e)}{\le}\sum_{k=1}^{n}I_{\mu_{k}}(\gamma_{b}),
\end{aligned}
\label{app_info_chain}
\end{equation}
where $(d)$ follows from the Markov chain $\varXi\to\mathbf{x}\to\mathbf{y}_{b}$ and $(e)$ uses the memoryless channel in~\eqref{scalar_ch}.
Dividing both sides of~\eqref{app_info_chain} by $J$ leads to
\vspace{-0.4em}
\begin{equation}
\begin{aligned}
\frac{1}{J}\sum_{j=1}^{J}I_{j}
&\le\frac{L}{n}\sum_{k=1}^{n}I_{\mu_{k}}(\gamma_{b})
\overset{(f)}{\le}L\,I_{\bar{\mu}}(\gamma_{b})\\
&\overset{(g)}{\le}L\,I_{\bar{\mu}}\big(\gamma_{b}^{\star}(\bar{\mu})\big)
\le L\,\eta_{n}^{\star},
\end{aligned}
\label{app_rate_chain2}
\end{equation}
where $(f)$ follows from the concavity of $\mu\mapsto I_{\mu}(\gamma)$ over the common constellation $\mathcal{X}$, $(g)$ follows from the SNR bound with the monotonicity of $I_{\bar{\mu}}$, and the last step follows from~\eqref{covert_sem_rate}.
The session-averaged distortion is then bounded as
\vspace{-0.4em}
\begin{equation}
\bar{D}_{J}
\overset{(h)}{\ge}\frac{1}{J}\sum_{j=1}^{J}\mathcal{D}_{\beta}(I_{j})
\overset{(i)}{\ge}\mathcal{D}_{\beta}\Big(\frac{1}{J}\sum_{j=1}^{J}I_{j}\Big)
\ge\mathcal{D}_{\beta}\big(L\,\eta_{n}^{\star}\big),
\label{app_distortion_chain}
\end{equation}
where $(h)$ follows from~\eqref{rd_function}, $(i)$ follows from the convexity of $\mathcal{D}_{\beta}$, and the last step follows from~\eqref{app_rate_chain2} with the monotonicity of $\mathcal{D}_{\beta}$.
Equality in $(a)$, $(c)$ and $(g)$ requires $\mathrm{TC}(\mathbb{P}^{(n)}_{1})=0$, $\mu_{1}=\dots=\mu_{n}$ and $\bar{P}=\bar{P}^{\star}$, which are guaranteed by~\eqref{target_law}.
This completes the proof of Theorem~\ref{theorem_distortion}.

\end{appendices}

\bibliographystyle{IEEEtran}
\bibliography{reference}

\end{document}